\documentclass[journal, twocolumn]{IEEEtran}     
\usepackage{amsmath} 
\usepackage{amsthm} 
\usepackage{graphicx} 
\usepackage{cite} 
\usepackage{amsmath,amsfonts}
\usepackage{algorithmic}
\usepackage{nccmath}
\usepackage{cases}
\usepackage{algorithm}
\usepackage{array}
\usepackage[caption=false,font=normalsize,labelfont=sf,textfont=sf]{subfig}
\usepackage{textcomp}
\usepackage{color}
\usepackage{stfloats}
\usepackage{url}
\usepackage{orcidlink}
\usepackage{makecell}
\usepackage{verbatim}
\usepackage{longtable}
\usepackage{graphicx}
\usepackage{cite}
\usepackage{cite}
\usepackage{amssymb}
\usepackage{blindtext}
\usepackage{amssymb}
\usepackage{mathrsfs}
\usepackage{amsmath}
\newtheorem{theorem}{Theorem}
\newtheorem{assertion}{Assertion}
\newtheorem{lemma}{Lemma}
\newtheorem{Corollary}{Corollary}

\newtheorem{Remark}{Remark}
\newtheorem{ob}{Observation}
\usepackage{float}
 
\begin{document}

\title{The dimension and Bose distance of certain primitive BCH codes}

\author{Run Zheng \orcidlink{0000-0002-9117-6639},  Nung-Sing Sze  \orcidlink{0000-0003-1567-2654}
 and Zejun Huang \orcidlink{0000-0003-2621-3234}
\thanks{Z. Huang's work was supported by the National Natural Science Foundation of China (No.12171323), Guangdong Basic and Applied Basic Research Foundation (No. 2022A1515011995).  N. S. Sze's work was supported by a HK RGC grant PolyU 15300121 and a PolyU
research grant 4-ZZRN. (\textit{Corresponding author: Zejun Huang})}
\thanks{R. Zheng is with the Department of Applied Mathematics, The Hong Kong Polytechnic University, Hung Hom, Hong Kong (e-mail: zheng-run.zheng@connect.polyu.hk).}

\thanks{N. S. Sze is with the Department of Applied Mathematics,
The Hong Kong Polytechnic University, Hung Hom, Hong Kong (e-mail:
raymond.sze@polyu.edu.hk).}

\thanks{Z. Huang is with the School of Mathematical Sciences,
    Shenzhen University, Shenzhen 518060, China (e-mail: zejunhuang@szu.edu.cn).}
}

\date{\today}


\maketitle
\begin{abstract}
BCH codes are a significant class of cyclic codes that play an important role in both theoretical research and practical applications. Their strong error-correcting abilities and efficient encoding and decoding methods make BCH codes widely applicable in various areas, including communication systems, data storage devices, and consumer electronics. Although BCH codes have been extensively studied, the parameters of BCH codes are not known in general. 
 Let $q$ be a prime power and $m$ be a positive integer.  Denote by $\mathcal{C}_{\left(q,m,\delta)\right)}$ the narrow-sense primitive BCH code with length $q^m-1$ and designed distance $\delta$. As of now, the dimensions of 
 $\mathcal{C}_{(q,m,\delta)}$ are fully understood only for $m \leq 2$. For $m \geq 4$, the dimensions of $\mathcal{C}_{(q,m,\delta)}$ are known only for 
 the range $2 \leq \delta \leq q^{\lfloor (m+1)/2 \rfloor +1}$ and for a limited number of special cases. In this paper, we determined the dimension and Bose distance of $\mathcal{C}_{(q,m,\delta)}$ for $m\geq 4$ and 
 $\delta\in [2, q^{\lfloor ( 2m-1)/{3}\rfloor+1}]. $ Additionally, we have also extended our results to some  primitive BCH codes that are not necessarily narrow-sense.

\end{abstract}

\section{Introduction}
\IEEEPARstart{T}{hroughout} this paper, let $q$ be a prime power, and   $\mathbb{F}_q$ be the finite field of order $q$. Let $\mathbb{F}_q^n$ be the $n$-dimensional linear space over $\mathbb{F}_q.$ A code of length $n$ over $\mathbb{F}_q$ is a nonempty subset of $\mathbb{F}_q^n$. 
An $[n,k,d]$  linear code $\mathcal{C}$ over $\mathbb{F}_q$ is a $k$-dimensional  subspace of $\mathbb{F}_q^n$  with  minimum distance $d$. A linear code $\mathcal{C}\subseteq \mathbb{F}_{q}^n$  is said to be \textit{cyclic} if $(c_{0},c_{1},\ldots,c_{n-1})\in \mathcal{C}$ implies that $(c_{n-1},c_{0},\ldots,c_{n-2})\in \mathcal{C}.$ 
Each vector $(c_0,c_1,\ldots,c_{n-1})\in  \mathbb{F}_q^n$ is identified with  its polynomial representation \begin{equation}\notag
c_0+c_1x+\cdots+c_{n-1}x^{n-1} \in \mathbb{F}_q[x]/(x^n-1),
\end{equation}
and  each  code $\mathcal{C}\subseteq \mathbb{F}_q^n$ is identified with  a subset of  the quotient ring $\mathbb{F}_q[x]/(x^n-1).$ In this way, a linear code $\mathcal{C}\subseteq \mathbb{F}_q^n$ is cyclic if and only if it is an ideal of the quotient ring $\mathbb{F}_q[x]/(x^n-1).$ 
 Note that each ideal of $\mathbb{F}_q[x]/(x^n-1)$ is principal.  Therefore, a cyclic  code $\mathcal{C}$ in $\mathbb{F}_q[x]/(x^n-1)$ can be denoted by $\mathcal{C}=\left \langle g(x) \right \rangle, $  where  $g(x)\in \mathcal{C}$ is  monic  with minimum degree.  Then  $g(x)$ is called the \textit{generator polynomial} of $\mathcal{C}$, and $h(x)=(x^n-1)/{g(x)}$ is called the \textit{ parity-check polynomial} of $\mathcal{C}$.

Suppose that $n$ is an integer such that $\mathrm{gcd}(n,q)=1$.  Denote  $m=\mathrm{ord}_n(q)$,  i.e., the smallest integer such that 
$q^m\equiv 1 \ (\mathrm{mod}\ n).$ Let $\alpha$ be a primitive element of $\mathbb{F}_{q^m}$.   Then $\beta=\alpha^{(q^m-1)/n} $ is a primitive $n$-th root of unity.  This leads to the  factorization  $x^n-1=\prod\limits_{i=0}^{n-1}(x-\beta^i).$  
For each integer $i\in [0,n-1]$,  we denote by $m_i(x)$ the minimal polynomial of $\beta^i$ over $\mathbb{F}_q$.  
A cyclic  code  of length  $n$  over $\mathbb{F}_q$ is called a  BCH code with designed distance $\delta$  if its generator polynomial has the form 
\begin{equation}\notag
\mathrm{lcm}(m_b(x),m_{b+1}(x),\ldots,m_{b+\delta-2}(x))
\end{equation}
for some integers  $b$ and $2\leq \delta\leq n$, 
where $\mathrm{lcm}$ denotes the least common multiple of the polynomials.     When $n=q^m-1$ for some positive  integer $m$,  such a code is  called a \textit{primitive} BCH code with designed distance $\delta$,   denoted  $\mathcal{C}_{(q,m, \delta, b)}$.             
If  $b=1$,  the code  $\mathcal{C}_{(q,  m, \delta, b)}$ is called  a  \textit{narrow-sense primitive}  BCH code,  denoted  $\mathcal{C}_{(q,m,\delta)}$.      Note that $\mathcal{C}_{(q,m,\delta,b)}$ and $\mathcal{C}_{(q,m,\delta^{'},b)}$ 
may be the identical for distinct $\delta$ and $\delta^{'}$.  The \textit{Bose distance} of $\mathcal{C}_{(q,m,\delta,b)}$, denoted  $d_B$ or $d_B(\mathcal{C}_{(q,m,\delta,b)})$, is the largest integer such that $\mathcal{C}_{(q,m,\delta,b)}=\mathcal{C}_{(q,m,d_B,b)} $.

Binary BCH codes were first independently discovered by Hocquenghem \cite{H1959} in 1959 and by Bose and Ray-Chaudhuri \cite{BR1960, BC1960f} in 1960. Gorenstein and Zierler \cite{GZ1961} later extended these binary BCH codes to $q$-ary BCH codes in 1961.  
BCH codes occupy a prominent place in theory and practice.  There exist many elegant
and powerful algebraic decoding algorithms for the BCH codes, including  Peterson–Gorenstein–Zierler algorithm,
Berlekamp–Massey algorithm,  and 
Sugiyama Euclidean algorithm \cite{B1968, M1969, P1960, GZ1961}. These algorithms are efficient, making BCH codes practical for real-world applications by
effectively correcting errors in transmitted data.   Due to these,  BCH codes are
widely used in communication systems, storage devices, and consumer electronics.  

{BCH codes have been extensively studied in the literature \cite{AASP2007,ADC1992, ADSN1994, CCGS2016, CPHTZ2006, DFZ2017, DYFT1997, FTT1986, HS1973, KILS2001, KOLS1999, KTFT1985, KTLS1972, LDL2017, LLDL2017, LRFL2019, RM1991, YH1996, CP1994,WXZ2024,FLD2023,WWL2023,XL2024,DLMQ2023,SLD2023,B1967, D2015, DDZ2015, NLMY2021, ZXP2024}. Despite decades of research, our understanding of the dimension and minimum distance of BCH codes remains limited. As noted by Charpin \cite{charpin1998open}, it is a challenging problem to determine the minimum distance $d$ of BCH codes in general. However, the BCH bound \cite{H1959, BR1960} indicates that $d\geq d_B$. The equality can even be obtained for some codes.   For example, the narrow sense BCH code $\mathcal{C}_{(q,m,\delta)}$ satisfy $d=d_B$ if $\delta$ divides $q^m-1$\cite[p. 259]{ma1977}.
Furthermore, it has been conjectured that $d\leq d_B+4$ for a narrow-sense primitive BCH code. Therefore, determining the Bose distance is also valuable for us to understand BCH codes.  Readers may refer to \cite{B1967, GLQ2024, D2015, DDZ2015, DFZ2017, cherchem2020some}  for some previous works on the Bose distance of $\mathcal{C}_{(q,m,\delta)}$. Determining the dimension of BCH codes is also difficult in general   \cite{ding2024bch,NLMY2021}. For narrow-sense primitive BCH codes $\mathcal{C}_{(q,m,\delta)}$, the dimension is fully determined for $m\leq 2$  \cite{LLDL2017,NLMY2021}.  For $m\geq 3$, the problem remains largely unsolved.  For convenience, we summarize some known results regarding the dimension of BCH codes of $\mathcal{C}_{(q,m,\delta)}$ in Table \ref{table}. For primitive BCH codes $\mathcal{C}_{(q,m,\delta,b)}$ that are not necessarily narrow-sense, there are quite limited results available. See \cite{liu2017} for the dimensions of some BCH codes that are not narrow-sense.   
The reader may also refer to \cite[Table 1]{pang2021five} and \cite{ding2024bch}  for summaries of the existing results on the parameters of BCH codes. }

{It is well known that the dimension and Bose distance of primitive BCH codes are closely related to $q$-cyclotomic coset modulo $q^m-1$. Specifically,  the dimension of BCH codes $\mathcal{C}_{(q,m,\delta)}$ can be given by calculating the sum of the sizes of distinct $ q$-cyclotomic cosets $C_a$ modulo $q^m-1$ for all $a\in [1,\delta-1]$. 
Therefore, 
to determine the dimension of the primitive BCH code $\mathcal{C}_{(q,m,\delta)}$,  it suffices to identify the size of $q$-cyclotomic coset $C_a$ for $ a\in [1,\delta-1]$, and enumurate distinct cosets correspond to each  size. Note that the size of a coset must be a divisor of $m$, which makes determining the size of any given coset relatively straightforward.  
However, finding the number of distinct cosets requires more effort. 
Yue et al. \cite{YH1996} introduced the concept of the coset leader, defined as the smallest integer in each coset. As a result, the number of distinct cosets corresponds exactly to the number of coset leaders. 
Moreover, determining the Bose distance of 
$\mathcal{C}_{(q,m,\delta)}$
  reduces to identifying the minimal coset leader greater than $\delta-1$. }

{The  $q$-adic expansion of the integer $a\in [0,q^m-1]$ provides a convenient way for determining  whether 
$a$ is the coset leader of $C_a$, which is a technique widely employed in the study of BCH codes. 
However,  determining the total number of coset leaders in $[b,b+\delta-2]$ is very difficult in general due to their highly irregular distribution. 
Only for sufficiently small $a$, the integer $a$ is a coset leader if and only if $q $ does not divide $a$.  Consequently, the number of coset leaders equals the number of integers not divisible by $q$ in the range $[b,b+\delta-2]$ when $b+\delta-2$ is small. This observation was applied in 
\cite{YH1996} to determine the dimension of $\mathcal{C}_{(q,m,\delta)}$ for
$ 2\leq \delta \leq  q^{\lfloor(m+1)/2\rfloor}+1$, and for $ q^{m/2} + 2\leq \delta \leq 2q^{m/2} + 1$ when $m$ is even. Later,  Liu et al. \cite{liu2017} improved this result by giving the dimension of $\mathcal{C}_{(q,m,\delta)}$ for $m\geq 4$ and $2\leq \delta\leq q^{\lfloor (2m-1)/3\rfloor+1}.$ Beyond this range, 
 the dimension is determined for a few special cases. For example, 
 Mann~\cite{M1962}  established the dimension of $\mathcal{C}_{(q,m,\delta)}$ for $\delta=q^t$. 
 Cherchem et al. \cite{cherchem2020some} give the dimension and Bose distance of $\mathcal{C}_{(q,m,\delta)}$ for $\delta=\frac{a(q^m-1)}{q-1}$ and $\delta=aq^m-1$ with $a\in [1,q-1]$.
By identifying the first few largest coset leaders and the sizes of corresponding cosets,  the parameters of $C_{(q,m,\delta)}$  for $\delta$ as the four largest coset leaders were determined in \cite{DFZ2017} \cite{li2018two} and \cite{MN2023}. 
One can also obtain the parameters of some  BCH codes by comparing them with those of codes with known parameters.    Building on Mann's work \cite{M1962}, Ding et al. determined the dimensions of 
 $\mathcal{C}_{(q,m,q^t-1)}$  \cite[Theorem 14]{D2015},
$\mathcal{C}_{(q,m,q^t+b)}$ 
for $1\leq b \leq \lfloor (q^t-1)/q^{m-t}\rfloor+1$ \cite[Theorem 13]{DDZ2015},  and $\mathcal{C}_{(q,m,q^t-2)}$ for $q>2$  \cite[Theorem 15]{D2015}. Recently, the dimension of $\mathcal{C}_{(q,m,q^t+b)} $ for   $\lceil{m}/{2}\rceil \leq t<m$ and $0\leq b<q^{m-t}+\sum\limits_{i=1}^{\lfloor {t}/{s}\rfloor-1}q^{m+i(m-t)}$, and $\mathcal{C}_{(q,m,q^t-b)} $ 
for $\lceil{m}/{2}\rceil < t<m$   and $0\leq b<(q-1)\sum\limits_{i=1}^{m-t}q^{i}$ were provided in \cite{GLQ2024}.  }

{In this paper, we determine the dimension and Bose distance of $\mathcal{C}_{(q,m,\delta)} $ for $m\geq 4$ and  $2\leq \delta\leq q^{\lfloor ( 2m-1)/{3}\rfloor+1}.$
We  also extend these results to  BCH codes $\mathcal{C}_{(q,m,\delta,b)} $ for $m\geq 4$ and positive integers $b$ and $\delta$ with $b+\delta\leq q^{\lfloor (2m-1)/3\rfloor+1}+1$. 
In the existing literature,  for $m\geq 4$, the dimension of $\mathcal{C}_{(q,m,\delta)}$ is  known only for $\delta\in [2,q^{\lfloor (m+1)/{2}\rfloor+1}]$ and for some special cases.  Moreover, it is clear that   $q^{\lfloor ( 2m-1)/{3}\rfloor+1}\geq q^{\lfloor (m+1)/{2}\rfloor+1}\cdot q^{\lfloor (m-4)/{6}\rfloor}$. Therefore,  our results provide the dimension of $\mathcal{C}
_{(q,m,\delta)}$ for $\delta$ in a much larger range than previously known.

Our results rely critically on the analysis of the distribution of coset leaders in  $[1, q^{\lfloor (2m-1)/3\rfloor+1}]$. 
Note that an integer $a\in [0,q^m-1]$  cannot be a coset leader if $q$ divides $a$. Therefore, to determine the number of coset leaders within a given range, we only need to count the integers that are neither divisible by $q$  nor coset leaders. We partitioned such integers in $\left[1, q^{\left\lfloor (2m-1)/3 \right\rfloor +1}\right]$ into distinct classes  with uniform representations by applying a one-to-one correspondence between the integer in  $[0,q^m-1]$ and  length-$m$ sequences of non-negative integers less than $q$. 
This classification allows us to compute the number of integers in each class through a uniform method. 
 Furthermore, we divide the interval $[1, \delta-1]$ into subintervals based on the distribution of coset leaders, so that we can establish an equation to express the number of integers neither divisible by $q$ nor coset leaders in each subinterval. Consequently, we derive the exact count of coset leaders in $[1,\delta-1]$ for each integer $\delta\in [2, q^{\lfloor (2m-1)/3 \rfloor +1}]$. 
 Combining this with our findings on the sizes of cosets, we ultimately determined the dimension of the code $\mathcal{C}_{(q,m,\delta)}$ for each integer $\delta\in [2, q^{\lfloor (2m-1)/3 \rfloor +1}]$. On the other hand, the Bose distance of 
$\mathcal{C}_{(q,m,\delta)}$
  is indeed equal to the smallest coset leader not less than $\delta$.
 Therefore, it is not difficult to determine the Bose distance of $\mathcal{C}_{(q,m,\delta)}$ based on our analysis of the distribution of coset leaders.}

\begin{table}[h]
\centering
\caption{Known results on the dimension of   $\mathcal{C}_{(q,m,\delta)}$}\label{table}
\begin{scriptsize}\begin{tabular}{|c|c|c|c|}
\hline
$m$ & $\delta$ & Year &  Reference\\ \hline
$m\geq 1 $ &$\delta=q^t$ & 1962  &  \cite{M1962}  \\ \hline
$m$ is odd & $2\leq \delta\leq q^{ ({m+1})/{2}}+1$ & 1996  &      \cite{YH1996}           \\  \hline
   $m$ is even  & $2\leq \delta\leq 2q^{{m}/{2}}+1$ & 1996  &     \cite{YH1996}          \\  \hline
$ m\geq 1$& \makecell[c]{$\delta=(q-l_0)q^{m-l_1-1}-1$ \\ with  $0\leq l_0\leq q-2$  \\ and  $0\leq l_1\leq m-1$} & 2015  &  \cite{D2015}  \\ \hline
$m$ is even & \makecell[c]{$\delta=k(q^{{m}/{2}}+1)$ \\ with $1\leq k\leq q-1 $} & 2015 &  \cite{D2015}  \\ \hline
$ m\geq 1$& $\delta= q^t-1$ & 2015  &  \cite{D2015}  \\ \hline
$ m\geq 1$&\makecell[c]{ $\delta=q^t+b$ with \\$1\leq b \leq \lfloor (q^t-1)/q^{m-t}\rfloor+1$} & 2015  &  \cite{DDZ2015}  \\ \hline
$m\geq 1$ & $\delta=(q-1)q^{m-1}-1-q^{\lfloor( {m-1})/{2} \rfloor}$  & 2017  &  \cite{DFZ2017}  \\ \hline 
$m\geq 4$ &  $\delta=(q-1)q^{m-1}-1-q^{\lfloor ({m+1})/{2} \rfloor}$ & 2017  &  \cite{DFZ2017}  \\ \hline 
\makecell{$m\geq 4$} & $2\leq \delta\leq q^{\lfloor({m+1})/{2}\rfloor+1}$ & 2017  &     \cite{liu2017}            \\
\hline
\makecell{$m=2$ } & $2\leq \delta\leq q^{m}-2$  & 2017  &      \cite{liu2017} \\
\hline
$m>1$ &\makecell[c]{$\delta=a({q^m-1})/{q-1}$ \\or $\delta=aq^{m-1}-1$ \\ with $1\leq a\leq q-1$} & 2020   &   \cite{cherchem2020some}  \\ \hline
$m\geq 11$&$\delta=(q-1)q^{m-1}-1-q^{\lfloor({m+3})/{2}\rfloor} $ & 2023  & \cite{MN2023}  \\ \hline
$m>1$ & \makecell[c]{$\delta=q^t+b$ with $\lceil{m}/{2}\rceil \leq t<m$ and \\ $0\leq b<q^{m-t}+\sum\limits_{i=1}^{\lfloor {t}/{s}\rfloor-1}q^{m+i(m-t)};$ \\
$\delta=q^t-b$ with $\lceil{m}/{2}\rceil < t<m$  \\ and $0\leq b<(q-1)\sum\limits_{i=1}^{m-t}q^{i}$
} & 2024    &  \cite{GLQ2024}  \\ \hline
\end{tabular}
\end{scriptsize}
\end{table}
This paper is organized as follows. 
In Section \uppercase\expandafter{\romannumeral2}, we present some preliminaries that reveal the close relationship between $q$-cyclotomic  cosets modulo $q^m-1$ and the dimension and Bose distance of BCH codes. 
Section \uppercase\expandafter{\romannumeral3} introduces some new results on $q$-cyclotomic cosets. 
 These findings are then utilized to determine the dimension of narrow-sense primitive BCH codes $\mathcal{C}_{(q,m,\delta)}$ for $2\leq \delta\leq q^{\lfloor ( 2m-1)/{3}\rfloor+1}$ in Section \uppercase\expandafter{\romannumeral4}. The result concerning the dimension of narrow-sense primitive BCH codes is subsequently extended to primitive BCH codes $\mathcal{C}_{(q,m,\delta, b)}$ for some positive integers $\delta$ and $b$ with integers $\delta +b \leq q^{\lfloor ( 2m-1)/{3}\rfloor+1}+1$ in Section \uppercase\expandafter{\romannumeral5}. Moreover, we also apply the findings on $q$-cyclotomic cosets to obtain the Bose distance for primitive BCH codes in Section \uppercase\expandafter{\romannumeral6}.
We then provide examples of the BCH codes studied in this paper and compare them with the tables of the best known linear codes maintained by Markus Grassl at http://www.codetables.de, which is called \textit{Database} later in this paper. Furthermore, in Section \uppercase\expandafter{\romannumeral7}, as an illustration of our main results, we present the dimension and Bose distance of 
$\mathcal{C}_{(q,m,\delta)}$ for $\delta=aq^{h+k}+b$ with integers $k\in [m-2h, \lfloor(2m-1)/3\rfloor-h ]$, $a\in [1,q-1]$ and $b\in \left[1,q^{m-h-k}\right]$, where $h=\lfloor m/2\rfloor $. 
  Finally, the conclusion of this paper is given in Section \uppercase\expandafter{\romannumeral8}.

\section{Preliminaries}
Suppose that $a$ is a real number.  Denote by $\lfloor a \rfloor$ the largest integer less than or equal to  $a$. 
In the following sections,  let $m$ be a positive integer,  $h=\lfloor {m}/{2}\rfloor$ and 
$n=q^m-1$.   
\begin{lemma}\label{pl1}\textnormal{\cite[Proposition 2.4]{Hou2018}}
Let  $\alpha\in \mathbb{F}_{q^m}$ and  $g(x)$ be the minimal polynomial of $\alpha$ over $\mathbb{F}_{q}$. Then 
$$g(x)=(x-\alpha)(x-\alpha^q)\cdots (x-\alpha^{q^{l-1}}),$$
 where
 $ l $ is the smallest integer such that $\alpha^{q^{l}}=\alpha.$ 
\end{lemma}
\begin{Remark}
It is well known that the integer $l$ divides $m$. 
\end{Remark}
For each integer $a\in [0,n-1]$,   the \textit{$q$-cyclotomic coset} of $a$ modulo $n$ is defined by 
\begin{equation}\label{def1}
C_a=\{ aq^{k} \  \mathrm{mod }\  n\mid k=0,\ldots,l_a-1\},
\end{equation}  where 
$l_a$ is the smallest integer such that
$aq^{l_a}\equiv a \ (\mathrm{mod}\ n) $. It is clear that the size of $C_a$ is equal to $l_a$.  The smallest integer in $C_a$  is called the  \textit{coset leader} of ${C}_a$. 
We can apply Lemma \ref{pl1}  to have 
\begin{equation}\notag
    m_a(x)=(x-\beta^a)(x-\beta^{aq})\cdots (x-\beta^{aq^{l_a-1}})= \prod\limits_{i\in C_a}(x-\beta^i),
\end{equation} 
where $l_a$ is the smallest integer such that $\beta^{aq^{l_a}}=\beta^a$,  equivalently, $aq^{l_a}\equiv a\  (\mathrm{mod}\ n).$ 
 Then the generator polynomial  $g(x)$ of $\mathcal{C}_{(q,m,\delta, b)}$ can be given by 
\begin{equation}\notag
g(x)=\mathrm{lcm}(m_b(x),m_{b+1}(x),\ldots,m_{b+\delta-2}(x))= \prod\limits_{i\in \mathcal{G}}(x-\beta^i)
\end{equation}
with $\mathcal{G}=\bigcup\limits_{a=b}^{b+\delta-2}C_a.$ 
It is clear that $\mathrm{deg}(g(x))=\left|\bigcup\limits_{a=b}^{b+\delta-2}C_a\right|$,  where $\mathrm{deg} $ denotes the degree of a polynomial,  and $|\cdot |$ denotes the size of a set.   
The dimension of $\mathcal{C}_{(q,m,\delta,b)}$ can be given by  $$\mathrm{dim}(\mathcal{C}_{(q,m,\delta,b)})=n-\left|\bigcup\limits_{a=b}^{b+\delta-2}C_a\right|.$$ 
 Notably,  one coset has a unique coset leader.  
  Therefore,  
\begin{equation}\label{01}
    \mathrm{dim}(\mathcal{C}_{(q,m,\delta)})=n-\sum\limits_{a\in \mathcal{L}(1,\delta)}\left|C_a\right|,
\end{equation}
where $\mathcal{L}(1,\delta)$ denotes the set of coset leaders in $[1,\delta-1].$ Additionally, for $\delta^{'}\geq \delta$, we can observe that  
\begin{equation}\notag
\bigcup\limits_{a=1}^{\delta-1}C_a = \bigcup\limits_{a=1}^{\delta^{'}-1}C_a
\end{equation} if and only if any  integer in $[\delta,\delta^{'}-1]$ is not a coset leader. Consequently, by recalling the definition of Bose distance $d_B$,  we have  
  \begin{equation}\label{Bose1}        
d_B(\mathcal{C}_{(q,m,\delta)})=\delta^{'},
 \end{equation}
 where  $\delta^{'}$ denotes 
 the smallest coset leader in $[\delta,n-1]$.

 For simplicity, we use  $a\mid b$ to denote that the integer $a$ divides the integer  $b$,   and $a\nmid b$ to denote that $a$ does not
  divide $b$. It is 
 well known that a necessary condition for an integer $a$ being the coset leader of $ C_a$  is $q\nmid a$. 
Therefore,  we only need to focus on finding those integers that are neither divisible by $q$ nor a coset leader. We define $\mathcal{S}$ as the set of all such integers in $[1,n-1]$. That is, \begin{IEEEeqnarray}{rcl}
\mathcal{S}&=&\{a\in [1,n-1] : q\nmid a\hbox{ and } a \hbox{ is not a coset leader}\}.\nonumber
\end{IEEEeqnarray}
Notably,  the following Lemma demonstrates that the condition  $q\nmid a$ is also sufficient for $a$ being a coset leader when $a$ is sufficiently small. 
  
\begin{lemma}\textnormal{\cite[Theorem 2.3]{YF2000}}\label{lll}
When  $m$ is an odd integer, for any integer  $a\leq  q^{h+1}$, $a$ is the coset leader of  $C_a$ if and only if $q\nmid a$. 
When $m$  is an even integer,  for any integer  $a\leq  2q^{h}$, $a$ is the coset leader of  $C_a$  if and only if $q\nmid a$.
\end{lemma}

\section{New results on \texorpdfstring{$q$}-cyclotomic  cosets}
In this section,  we present some new results concerning the sizes and the coset leaders of $q$-cyclotomic  cosets. 
\begin{theorem}\label{th1}
When  $m$ is an  odd integer,  for  any integer $a\in [1,q^{m-\lfloor {m}/{3}\rfloor}),$
\begin{equation}\notag
|C_a|=m. 
\end{equation}
When  $m$ is an even integer,   for any integer  $a\in  [1,q^{m-\lfloor{m}/{3}\rfloor}),$ 
\begin{IEEEeqnarray}{c}
|C_a|=\left\{ \,
\begin{IEEEeqnarraybox}[][c]{l?s}
\IEEEstrut
\frac{m}{2} & if  $aq^{h}\mathrm{\ mod\ } n = a$, \\
m & if $aq^{h} \mathrm{\ mod\ }n \neq  a$. 
\IEEEstrut
\end{IEEEeqnarraybox}
\right.
\label{th2e}
\end{IEEEeqnarray}
\end{theorem} 
In order to prove the above theorem,  we introduce some additional notations and concepts. Let  $Z_q$ be the set of all non-negative integers less than $q$. Each integer $a\in [0, n]$ can be uniquely represented by its \textit{$q$-adic expansion} as $a=\sum\limits_{\ell=0}^{m-1}a_{\ell}q^{\ell}$ with
$a_{\ell}\in Z_q.$
Let $Z_q^{m}$ be  the set of all length-$m$
sequences of integers in  $Z_q$.   For simplicity, denote by $\mathbf{0}_{m}$ the sequence in $Z_q^m$ whose elements are all zero. 
We define an order on $Z_q^m$ using lexicographic order. Specifically, for any two squences   $U=(u_{m-1},\ldots,u_1,u_{0})$ and $W=(w_{m-1},\ldots,w_1,w_0)$ in $Z_q^m$,  
\begin{itemize}
\item [1.] $U$ and $W$  are said to be equal, denoted as $U=W$,  if    $u_{\ell}=w_{\ell}$ for $\ell=0,\ldots,m-1,$
\item  [2.] $U$ is less than $W$, denoted as $U<W$, if either  $u_{m-1}<w_{m-1}  $  or there exists an integer $i\in [0,m-2]$ such that  
$u_i<v_i$ and $u_{\ell}=w_{\ell}$ for all $\ell=i+1,\ldots,m-1,$ and
\item [3.] $U\leq W$ is denoted  if $U= W$ or $U<W$.
\end{itemize}  
We define a map  $V$ from the set of all the integers in $[0, n]$ to $Z_q^{m}$  as 
  \begin{equation}\notag
  V(a)=(a_{m-1}, \ldots,a_1, a_0),
\end{equation}   
where   $\sum\limits_{\ell=0}^{m-1}a_{\ell}q^{\ell}$ forms the  $q$-adic expansion of the integer  $a\in [0,n].$ 
Let $I_{m-1}$ be the identity matrix of  order $m-1$, and let   
\begin{equation}\notag
    Q=\begin{bmatrix}
        0& 1\\
        I_{m-1}& 0
    \end{bmatrix}.
\end{equation}
It is clear that for any squence $(a_{m-1},a_{m-2},\ldots,a_0)\in Z_q^m$,  
\begin{equation}\notag
 (a_{m-1},a_{m-2},\ldots,a_0)Q=(a_{m-2},\ldots,a_{0},a_{m-1}).   
\end{equation}
Suppose that 
 $a, b\in [0,n-1]$ are two integers.  We  have the following observations:
\begin{ob}\label{ob1}
  $a=b$  if and only if $V(a)=V(b)$, $a<b$  if and only if $V(a)<V(b)$,   and $a\leq b$  if and only if  $V(a)\leq V(b)$.
\end{ob}
\begin{ob}\label{ob2}
 $V(aq^t \mathrm{\ mod\ }n)=V(a)Q^t $ for any integer $t\in [1,m-1].$
 \end{ob}
\begin{ob}\label{ob3}
The integer $a$ is the coset leader of $C_a$ if and only if 
$V(a)\leq V(a)Q^{t}$ for all integer $t\in [1,m-1].$ 
\end{ob}
\textbf{Proof of Theorem \ref{th1}.}
We first consider the case when $m$ is an odd integer. 
Suppose that  $a\in [1,q^{m-\lfloor{m}/{3}\rfloor})$ is an integer.
Denote by $\sum\limits_{\ell=0}^{m-1}a_{\ell}q^{\ell}$ the $q$-adic expansion of $a$.  It is clear that  
$a_{\ell}=0$ for all $\ell=m-\lfloor {m}/{3}\rfloor, \ldots, m-1$, and hence 
\begin{equation}\label{th1e2}
V(a)=( \mathbf{0}_{\lfloor\frac{m}{3}\rfloor}, a_{m-\lfloor\frac{m}{3}\rfloor-1},\ldots,a_0 ).
\end{equation}
By the definition of the coset $C_a$, we have 
$aq^{l_a} \mathrm{\ mod\ } n=a$.
It follows that
$V(a)Q^{l_a}=V(aq^{l_a} \mathrm{\ mod\ }n)=V(a),$ i.e., 
   \begin{equation}\notag
   (a_{m-1-l_a},\ldots,a_0,a_{m-1},\ldots,a_{m-l_a})=(a_{m-1},\ldots,a_1,a_0).
   \end{equation}
   For simplicity, denote $U=(a_{m-1},\ldots,a_{m-l_a}).$ Note that $l_a$ divides $m$. 
Then we can conclude from the above equation that 
\begin{equation}\label{th1e3}
V(a)=(U, U,\ldots,  U). 
\end{equation}  
Assume that $l_a \leq \lfloor {m}/{3}\rfloor$. 
We can conclude from  (\ref{th1e2}) 
that $U=(a_{m-1},\ldots,a_{m-\ell_a})=\mathbf{0}_{\ell_a}.$ With (\ref{th1e3}), it follows   that $V(a)=\mathbf{0}_m$, 
which implies $a=0$. This contradicts the initial assumption that $a\in [1,q^{m-\lfloor {m}/{3}\rfloor})$. Therefore, we must have $l_a>\lfloor {m}/{3}\rfloor$. Given that $m$ is an odd integer and $l_a$ divides $m$, it follows that $l_a=m$. 
Equivalently,  $|C_a|=m$.

 When $m$ is an even integer,  we can use a similar argument as in the first paragraph to show that $l_a>\lfloor {m}/{3}\rfloor$ for any integer $a\in  [1,q^{m-\lfloor {m}/{3}\rfloor}).$  Note that $l_a$ divides $m$ and $m$ is even.  It follows that $l_a={m}/{2}$ or $m$. 
 Consequently, we have $l_a={m}/{2}=h$
  if   $aq^{h}\mathrm{\ mod\ } n = a$, and $l_a=m$ otherwise. 
Equivalently, the equality in  (\ref{th2e}) holds.   This completes the proof.  \qed

It is straightforward to verify that  $aq^{h}\mathrm{\ mod\ } n \neq a $ for any integer $a\in [1, q^h)$ when $m$ is even. Therefore, we can derive the following corollary, which can also be derived from \cite[Lemma 5]{YH1996}.  
\begin{Corollary}
\label{cos}   If $m$ is an even integer,  then $|C_a|=m $ for any integer $a\in [1,q^h).$
\end{Corollary} 
Let  $\mathcal{H}$ be the set of all integers $a\in [1,n-1]$ such that $a$ is a coset leader and $|C_a|={m}/{2}$. That is, 
\begin{equation}\notag
    \mathcal{H}=\left\{a\in [1,n-1]: a \hbox{ is a  coset leader and } |C_a|=\frac{m}{2}\right\}.
\end{equation}
 Then we also  the following  corollary of  Theorem \ref{th1}. 
\begin{Corollary}\label{corr}
Let $m$ be an even integer and  $k$ be an integer such that  $0\leq k\leq \lfloor(2m-1)/{3}\rfloor-h$. Suppose that $ a\in  \left[q^{h+k},q^{h+k+1}\right)$   is  an integer. Then   $a\in \mathcal{H}$ if and only if $V(a)$ has the form 
\begin{equation}\label{coeq1}
(\mathbf{0}_{h-k-1},a_{k},\ldots,a_0, \mathbf{0}_{h-k-1},a_{k},\ldots,a_0)
\end{equation}
with  $a_0>0$ and $a_k>0$.   
\end{Corollary}

\begin{proof}
\textit{Necessity part.} Suppose that $a\in \mathcal{H}$. By definition of $\mathcal{H}$ and  applying Theorem 1, we  conclude that $q\nmid a$ and  $aq^{h}\mathrm{\ mod\ } n = a.$ Note that the latter condition is equivalent to $V(a)Q^{h}=V(a)$, and   $V(a)$ has the form 
$
(\mathbf{0}_{h-k-1},a_{h+k}\ldots, a_0)
$ with $a_{h+k}>0$. 
It follows that 
\begin{multline}\notag
(\mathbf{0}_{h-k-1},a_{h+k}\ldots, a_0) \\=   (a_{h-1},\ldots,a_0, \mathbf{0}_{h-k-1},a_{h+k}\ldots, a_h). 
\end{multline}
This implies that \begin{equation}\notag
    (a_{h-1},\ldots,a_{k+1})=\mathbf{0}_{h-k-1}
\end{equation}
and \begin{equation}\notag
    (a_{k},\ldots,a_0)=(a_{h+k},\ldots,a_h).
\end{equation}
Therefore, $V(a)$ has the form specified in (\ref{coeq1}). Furthermore, since  $q\nmid a$ and $a_{h+k}>0$, we also have $a_0>0$ and $a_{k}=a_{h+k}>0$. 

\textit{Sufficiency part.} 
Suppose that $V(a)$ has the form specified in (\ref{coeq1}) with $a_0>0$ and $a_k>0$. Then it is straightforward to verify that $V(a)\leq V(a)Q^t$ for any integer $t\in [1, m-1]$. In particular,  $V(a)=V(a)Q^h$, which is equivalent to $aq^{h}\mathrm{\ mod\ } n = a.$ 
By applying  Theorem \ref{th1}  and Observation \ref{ob3}, we can conclude that    $a$ is the coset leader of $C_a$ and  $|C_a|={m}/{2}$, i.e., $a\in \mathcal{H}$. 
\end{proof}
By counting the sequences in $Z_q^m$ having the form in (\ref{coeq1}) with $a_{0}>0$ and $a_k>0$,  we can directly derive the following corollary. 
\begin{Corollary}\label{h}
Let $m\geq 4$ be an even integer and  $k$ be an integer such that  $1\leq k\leq \lfloor(2m-1)/{3}\rfloor-h$. Then we have
\begin{equation}\label{coeq2}
\left|   \left[q^{h+k},q^{h+k+1}\right)\cap \mathcal{H}\right|=
   q^{k-1}(q-1)^{2}.  
\end{equation}
\end{Corollary}

Let $m$ and $k $  be two integers such that $m\geq 4$ and  $m-2h\leq k\leq  \lfloor {(2m-1)}/{3}\rfloor-h$. 
When $m$ is an odd integer,  
  for each integer $i\in [-k+1, k]$, we define  $\mathcal{A}_k(i)$ as  the set of  all  integers $a\in [q^{h+k},q^{h+k+1})$ with $q$-adic expansion $\sum\limits_{\ell=0}^{h+k}a_{\ell}q^{\ell}$ that satisfies: 
\begin{IEEEeqnarray}{l}
a_{h+i}>0; \label{p31}\\
(a_{k+i-1},\ldots,a_{0})\leq (a_{h+k}, \ldots, a_{h-i+1}) \hbox{ and }a_0>0;  \label{p32} \\ 
V(a)=(\mathbf{0}_{h-k}, a_{h+k},\ldots,a_{h+i},\mathbf{0}_{h-k},a_{k+i-1},\ldots,a_0).  \label{p33}
 \IEEEeqnarraynumspace 
\end{IEEEeqnarray}
When $m$ is an even integer,  for each $i\in [-k,k]$,  we define  $\mathcal{B}_k(i)$ as the set of all integers $a\in [q^{h+k},q^{h+k+1})$  with $q$-adic expansion $\sum\limits_{\ell=0}^{h+k}a_{\ell}q^{\ell}$ that satisfies the condition in (\ref{p31}) and the following: 
\begin{IEEEeqnarray}{l}
 (a_{k+i},\ldots,a_{0})\leq (a_{h+k}, \ldots, a_{h-i}) \hbox{ with }a_0>0;  \label{p35} \\
V(a)\!=\!(\mathbf{0}_{h-k-1}, a_{h+k},\ldots,a_{h+i},\mathbf{0}_{h-k-1},a_{k+i},\ldots,a_0).  \label{p36}   \IEEEeqnarraynumspace
\end{IEEEeqnarray}
We make the following remarks about the above definitions of $\mathcal{A}_k(i)$ and $\mathcal{B}_k(i)$. 
\begin{Remark}\label{rr2}
  The condition   that $k\leq \lfloor({2m-1})/{3}\rfloor-h$ implies the following consequences:
 \begin{itemize}
     \item When $m$ is odd, $|2i-1|\leq h-k$ for all $i\in [-k+1,k]$. 
\item  When $m$ is even, $|2i|\leq h-k-1$  for all $i\in [-k, k]$.
  \end{itemize} 
\end{Remark}   
\begin{Remark}\label{rm3}
 The condition in (\ref{p31})  can be equivalently represented as $q\nmid \sum\limits_{\ell=h+i}^{h+k}a_{\ell}q^{\ell-(h+i)}$.
\end{Remark}
\begin{Remark}\label{rm4}
 The  condition  in    (\ref{p32}) can be equivalently expressed as  \begin{equation}\label{mm}
     \sum\limits_{\ell=0}^{k+i-1}a_{\ell}q^{\ell}\leq \sum\limits_{\ell=h-i+1}^{h+k}a_{\ell}q^{\ell-(h-i+1)}\quad \hbox{and}\quad a_0>0.
 \end{equation}
Moreover,  the form of $V(a)$ in (\ref{p33}) implies that 
\begin{equation}\notag
    \sum\limits_{\ell=h-i+1}^{h+k}a_{\ell}q^{\ell-(h-i+1)}=\sum\limits_{\ell=h+i}^{h+k}a_{\ell}q^{\ell-(h+i)}\cdot q^{2i-1} 
\end{equation} 
 if $i\in [1,k],$ and 
\begin{IEEEeqnarray}{rCl}
\sum\limits_{\ell=h+i}^{h+k}a_{\ell}q^{\ell-(h+i)}\cdot q^{2i-1}&=& 
\sum\limits_{\ell=h-i+1}^{h+k}a_{\ell}q^{\ell-(h-i+1)} \nonumber \\ \nonumber &&+\>   \sum\limits_{\ell=h+i}^{h-i}a_{\ell}q^{\ell-(h-i+1)} \nonumber
\end{IEEEeqnarray} 
if $i\in  [-k+1,0]$. 
Noticing that  $0\leq \sum\limits_{\ell=h+i}^{h-i}a_{\ell}q^{\ell-(h-i+1)}<1$ for $i\in [-k+1,0]$, and the two summations in (\ref{mm}) are both 
integers, we can conclude that  (\ref{mm}) is equivalent to  
\begin{equation}\label{uu}
   \sum\limits_{\ell=0}^{k+i-1}a_{\ell}q^{\ell}\leq    \sum\limits_{\ell=h+i}^{h+k}a_{\ell}q^{\ell-(h+i)}\cdot q^{2i-1}\quad \hbox{and}\quad a_0>0.
\end{equation}
Similarly, 
the condition in  (\ref{p35}) can also   be written as  
 \begin{equation} \sum\limits_{\ell=0}^{k+i}a_{\ell}q^{\ell}\leq \sum\limits_{\ell=h-i}^{h+k}a_{\ell}q^{\ell-(h-i)} \quad\hbox{and}\quad a_0>0,\end{equation} which is further equivalent to 
 \begin{equation} \sum\limits_{\ell=0}^{k+i}a_{\ell}q^{\ell}\leq \sum\limits_{\ell=h-i}^{h+k}a_{\ell}q^{\ell-(h+i)}\cdot q^{2i}\quad\hbox{and}\quad a_0>0.
 \end{equation}
\end{Remark}

\begin{Remark}\label{rm6}
If   $a\in \mathcal{A}_k(i)$ is an integer with  $q$-adic expansion $\sum\limits_{\ell=0}^{h+k}a_{\ell}q^{\ell},$ then 
$\sum\limits_{\ell=0}^{h-k}a_{\ell}q^{\ell}\leq \sum\limits_{\ell=h-i+1}^{h+k}a_{\ell}q^{\ell-(h-i+1)}$.

Similarly,  if $a\in \mathcal{B}_k(i)$ is an integer with  $q$-adic expansion $\sum\limits_{\ell=0}^{h+k}a_{\ell}q^{\ell}$,  then   
$\sum\limits_{\ell=0}^{h-k-1}a_{\ell}q^{\ell}\leq \sum\limits_{\ell=h-i}^{h+k}a_{\ell}q^{\ell-(h-i)}$. 
\end{Remark}

\begin{Remark}\label{cor2}
If   $a\in \mathcal{A}_k(i)$ is an integer with  $q$-adic expansion $\sum\limits_{\ell=0}^{h+k}a_{\ell}q^{\ell},$ then 
$i$ is the smallest integer in $[-k+1,k]$ such that $a_{h+i}>0$.  Consequently, $\mathcal{A}_k(i)\cap \mathcal{A}_k(j)=\varnothing$  for  distinct integers $i,j\in [-k+1,k]$. 

Similarly,  if $a\in \mathcal{B}_k(i)$ is an integer with  $q$-adic expansion $\sum\limits_{\ell=0}^{h+k}a_{\ell}q^{\ell}$,  then   $i$ is the smallest integer in $[-k,k]$ such that $a_{h+i}>0$.  As a result,  $\mathcal{B}_k(i)\cap \mathcal{B}_{k}(j)=\varnothing$ for distinct integers $i,j\in [-k,k]$. 
\end{Remark}

We denote by $\bigsqcup\limits_{i\in I} A_i$ the disjoint union of a family of sets $\{A_i: i\in I\}$, where $I$ is an index set. Then we have the following Theorem.
\begin{theorem}\label{th2}
Let $m$ and $k$ be  two integers such that $m\geq 4$ and  $m-2h\leq k\leq   \lfloor (2m-1)/{3}\rfloor -h.$\begin{itemize}
    \item 
If $m$ is odd,  then 
 \begin{equation}\label{th2e1}
\mathcal{S}\cap [q^{h+k},q^{h+k+1})=\bigsqcup\limits_{i=-k+1}^k\mathcal{A}_k(i).
\end{equation}
\item If $m $ is  even,   
 then \begin{IEEEeqnarray}{rCl}
(\mathcal{S}\cup \mathcal{H})\cap [q^{h+k},q^{h+k+1}) =\bigsqcup\limits_{i=-k}^k\mathcal{B}_k(i).  \label{th2e3}  
\end{IEEEeqnarray}
\end{itemize}
\end{theorem}

\begin{proof}
We first consider the case when   $m$ is an odd integer.  
Suppose that  $a\in \bigcup\limits_{i=-k+1}^k\mathcal{A}_k(i)$ is an integer with $q$-adic expansion $\sum\limits_{\ell=0}^{h+k}a_{\ell}q^{\ell}$. 
By the definition of  $\mathcal{A}_k(i)$,   we have   (\ref{p31})   (\ref{p32}) and (\ref{p33}) 
for some integer $i\in [-k+1,k]$. 
It follows that $q\nmid a$ and \begin{multline}
    \label{www1}
    V(a)Q^{h-i+1} \\= (\mathbf{0}_{h-k}, a_{k+i-1},\ldots,a_0, \mathbf{0}_{h-k},a_{h+k},\ldots,a_{h+i}).
\end{multline}

If $i\in [-k+1, 0]$,  then we can conclude  from  (\ref{p31})  and (\ref{p32}) that  \begin{equation}\label{pth1}
(a_{k+i-1},\ldots,a_0,\mathbf{0}_{-2i+1})<(a_{h+k},\ldots,a_{h+i}).
\end{equation}
Recalling  the fact that   $|2i-1|\leq h-k$ in Remark \ref{rr2},  
 it follows from (\ref{www1}) and (\ref{pth1}) that $ V(a)Q^{h-i+1}<V(a).$ By the definition of $\mathcal{S}$ and Observation \ref{ob3},  we have   $a\in \mathcal{S}$.  
On the other hand, 
if   $i\in [1, k]$,  
we conclude from  $|2i-1|\leq h-k$ and  (\ref{p33}) that 
\begin{IEEEeqnarray}{rCl} 
(a_{h+k}, \ldots, a_{h-i+1})=(a_{h+k}, \ldots, a_{h+i},\mathbf{0}_{2i-1}).\IEEEeqnarraynumspace\nonumber
\end{IEEEeqnarray}
With  (\ref{p32}),  it follows that  
\begin{equation}\notag
    (a_{k+i-1},\ldots,a_{2i-1})<(a_{h+k}, \ldots, a_{h+i}).
\end{equation}
Consequently,  we still have 
$V(a)Q^{h-i+1}<V(a)$, 
which implies $a\in \mathcal{S}$.  
By  now, we have already demonstrated  that  $\bigcup\limits_{i=-k+1}^k\mathcal{A}_k(i)\subseteq  \mathcal{S}\cap [q^{h+k},q^{h+k+1}).$

Conversely, let us assume  that $a\in  \mathcal{S}\cap [q^{h+k},q^{h+k+1}) $ is an integer with $q$-adic expansion 
$\sum\limits_{\ell=0}^{h+k}a_{\ell}q^{\ell}$. We first have $a_{h+k}>0$. 
By the definition of $\mathcal{S}$ and Observation \ref{ob3}, we also  have  $a_0>0$ and  $V(a)Q^{t}<V(a)$  for some integer $t\in [1, m-1]$. 
Noticing that $V(a)=(\mathbf{0}_{h-k}, a_{h+k},\ldots,a_0)$, it follows that   the first  $h-k$ entries of $V(a)Q^t$ are all equal to zero. Considering $a_{h+k}>0$ and $a_{0}>0$, we  have  $h-k+1\leq t\leq h+k$.   Consequently, we can conclude that  \begin{equation}\notag
    V(a)Q^t=(\mathbf{0}_{h-k}, a_{h+k-t},\ldots, a_0,\mathbf{0}_{h-k}, a_{h+k},\ldots,a_{m-t})
\end{equation}
and \begin{equation}\notag
    V(a)=(\mathbf{0}_{h-k}, a_{h+k},\ldots, a_{m-t},\mathbf{0}_{h-k}, a_{h+k-t},\ldots,a_{0}).
\end{equation}
Given that  $a_{h+k}>0$,  we can chose $i$ to be the smallest integer in  $[m-h-t, k]$ such that $a_{h+i}>0$. Then 
$V(a)$ has the form specified in (\ref{p33}) with $(a_{k+i-1},\ldots,a_{0})=(\mathbf{0}_{h+t+i-m}, a_{h+k-t},\ldots,a_0)$. Furthermore, the inequality $h-k+1\leq t\leq h+k$ implies that $i\in [-k+1,k]$. 

Additionally,
if $i>m-h-t$, we can obtain (\ref{p32}) from the fact that $a_{h+k}>0$ and $a_0>0$.  On the other hand, if $i=m-h-t$, then we can derive (\ref{p32}) from $V(a)Q^t<V(a)$ and $a_0>0$. Now we
can conclude that $a\in \mathcal{A}_k(i)$ for some integer $i\in [-k+1,k]$. It follows that 
 $\mathcal{S}\cap [q^{h+k},q^{h+k+1})\subseteq \bigcup\limits_{i=-k+1}^k\mathcal{A}_k(i).$

We now  can conclude from the above argument that $\bigcup\limits_{i=-k+1}^k\mathcal{A}_k(i)= \mathcal{S}\cap [q^{h+k},q^{h+k+1}).$ 
Finally,  applying  Remark \ref{cor2} directly yields equation \eqref{th2e1}. 

If $m$ is an even integer,  with Corollary \ref{corr}, we can utilize a similar argument as above to show that $$\mathcal{S}\cap [q^{h+k},q^{h+k+1})=\bigcup\limits_{i=-k}^k\mathcal{B}_k(i)\setminus \left[\mathcal{H}\cap  [q^{h+k},q^{h+k+1})\right].$$ 
It follows that $(\mathcal{S}\cup \mathcal{H})\cap [q^{h+k},q^{h+k+1})
 =\bigcup\limits_{i=-k}^k\mathcal{B}_k(i).$
Then by applying Remark \ref{cor2} again, we have equation (\ref{th2e3}). This completes the proof of Theorem \ref{th2}.  
\end{proof}

\section{The dimension of \texorpdfstring{$\mathcal{C}_{(q,m,\delta)}$ } \ for \texorpdfstring{$2 \leq \delta \leq q^{\lfloor ({2m-1})/{3}\rfloor+1} $ } \
}

Let $\mathbb{Z}$ denote the set of all the integers. For a real number $a$,  denote by $N(a)$ the number of integers in the range $[1, a-1]$  not divisible by $q$, i.e.,   $N(a)=\left \lfloor a-1\right \rfloor- \left \lfloor (a-1)/q \right \rfloor$.  
Let $m\geq 4$ and   $\delta\in [2, q^{\lfloor (2m-1)/{3}\rfloor +1}]$ be  integers. Let $k_{\delta}=\lfloor \log_q(\delta-1)\rfloor-h$.  This implies that    $q^{h+k_{\delta}}\leq \delta-1<q^{h+k_{\delta}+1}$. Let  
$\sum\limits_{\ell=0}^{h+k_{\delta}}\delta_{\ell}q^{\ell}$ be the $q$-adic expansion of $\delta-1$. If  $k_{\delta}\geq m-2h$,   let $s_{\delta}$ denote the smallest integer in $[m-2h-k_{\delta}, k_{\delta}]$ such that $\delta_{h+s_{\delta}}>0$. 

If $m$ is odd, we define the function $f(\delta)$ for integers  $\delta\in [2, q^{\lfloor (2m-1)/{3}\rfloor +1}]$ as follow:
\begin{equation} \label{ff}
f(\delta)=
\begin{cases}                  
0,  &  \delta \leq q^{h+1}, \\
 \left(\begin{aligned}
 &q^{2k_{\delta}-3}(k_{\delta}-1)(q-1)^2\\
 &+N\left(\mu(\delta)+1\right)   \\ 
&+\sum\limits_{i=-k_{\delta}+1}^{k_{\delta}}\sum\limits_{t\in \mathcal{T}_i(\delta)}N(tq^{2i-1}+1)
\end{aligned}\right), &  \delta>q^{h+1},  
\end{cases}
\end{equation}
where 
$\mathcal{T}_i(\delta)= \left\{t\in \mathbb{Z}: q^{k_{\delta}-i}\leq  t <\sum\limits_{\ell= h+s_{\delta}}^{h+k_{\delta}}\delta_{\ell}q^{\ell-h-i},
q\nmid t\right\} $  
 and 
$ \mu(\delta)= \min
\left \{\sum\limits_{\ell=0}^{h-k_{\delta}}\delta_{\ell}q^{\ell}, \sum\limits_{\ell=h-s_{\delta}+1}^{h+k_{\delta}}\delta_{\ell}q^{\ell-(h-s_{\delta}+1)}\right \}.$

If $m$ is even,  we define the  function $\widetilde{f}(\delta)$ for integers $\delta\in [2, q^{\lfloor (2m-1)/{3}\rfloor +1}]$ as follow:
\begin{equation} \label{f}
\widetilde{f}(\delta)=
\begin{cases}
0, \hspace{-10pt}  & \hspace{-6pt}  \delta \leq q^{h}, \\
\frac{1}{2}(\delta_h-1)\delta_h+N(\widetilde{\mu}(\delta)+1), \hspace{-10pt} & \hspace{-6pt}q^h<\delta \leq q^{h+1},\\
 \left(\begin{aligned}
& (k_{\delta}-\textstyle\frac{1}{2})q^{2k_{\delta}-2}(q-1)^2  \\
&+\textstyle\frac{1}{2}q^{k_{\delta}-1}(q-1)\\ &+N(\widetilde{\mu}(\delta)+1) \\ 
&   +  \sum\limits_{i=-k_{\delta}}^{k_{\delta}}\sum\limits_{t\in {\mathcal{T}}_i(\delta)}N(tq^{2i}+1)
\end{aligned}\right), \hspace{-10pt} & \hspace{-6pt} \delta >q^{h+1},
\end{cases}
\end{equation}
where 
$\widetilde{\mu}(\delta)=\min\left\{ \sum\limits_{\ell=0}^{h-k_{\delta}-1}\delta_{\ell}q^{\ell}, \sum\limits_{\ell=h-{s_{\delta}}}^{h+k_{\delta}}\delta_{\ell}q^{\ell-(h-{s_{\delta}})} \right\}.$ We define the function $\tau(\delta)$ for integers $\delta \in (q^{h}, q^{\lfloor (2m-1)/{3}\rfloor +1}]$ as follow:  
\begin{equation}\notag
    \tau(\delta)=
    \begin{cases}
1, &\hbox{if } \delta_h>0 \hbox{ and } 
\sum\limits_{\ell=h}^{h+k_{\delta}}\delta_{\ell}q^{\ell-h}\leq \sum\limits_{\ell=0}^{h-1}\delta_{\ell}q^{\ell}, \\
0, & \hbox{otherwise. } 
\end{cases}
\end{equation}
In addition, we  define the function $g(\delta)$ for integers  $\delta\in [2, q^{\lfloor (2m-1)/{3}\rfloor +1}]$ as follow:
\begin{equation}\label{g}
    g(\delta)=
    \begin{cases}
0, &  \delta\leq q^h,\\
\delta_h-1+\tau(\delta), & q^h<\delta\leq q^{h+1},\\
N\left(\sum\limits_{\ell=h}^{h+k_{\delta}}\delta_{\ell}q^{\ell-h}\right)+\tau(\delta), & \delta >q^{h+1}.
\end{cases}
\end{equation}

\begin{theorem}\label{odda}
Let $m\geq 4$ and $\delta \in \left[2, q^{\lfloor (2m-1)/{3}\rfloor +1}\right]$ be integers. 
\begin{itemize}
    \item 
 If $m$ is odd,  then 
\begin{equation}\label{th31}
\mathrm{dim}(\mathcal{C}_{(q,m,\delta)})=
n-m\left[N(\delta)-f(\delta)\right]. 
\end{equation}
\item 
If $m$ is  even,  then  
\begin{equation} \label{th311}
\mathrm{dim}(\mathcal{C}_{(q,m,\delta)})=
n-m\left[N(\delta)-\widetilde{f}(\delta)\right]-\frac{m}{2}g(\delta).
\end{equation}
\end{itemize}
\end{theorem}

    We prove Theorem \ref{odda} through Assertions \ref{as1} -- \ref{as5}, which are established using the following lemmas. Complete proofs of lemmas  \ref{le11} and \ref{lemma10} appear in Appendices A to B, while detailed proofs of Assertions \ref{as1} -- \ref{as5} can be found in Appendices C to F.

\begin{lemma}\label{le5}
Let $m$ be a positive integer and  $t<q^m$ be a real number. Then 
\begin{multline}\notag
\left|\left \{(a_{m-1},\ldots,a_0)\in Z_q^{m}: \sum\limits_{\ell=0}^{m-1}a_{\ell}q^{\ell}\leq t,   a_0>0\right \}\right| \\
=N(t+1).
\end{multline}
\end{lemma}
\begin{proof}
There exists a one-to-one correspondence between the sequences in the above set
and the integers in $[1,t]$ that are not divisible by $q$ by mapping $(a_{m-1},\ldots,a_0)$ to $\sum\limits_{\ell=0}^{m-1}a_{\ell}q^{\ell}$. Therefore, the result follows.  
\end{proof}

\begin{lemma}\label{le11}
Let $k\geq 0$ be an integer.  
Then 
\begin{equation}\notag
\sum\limits_{t=q^{k}}^{q^{k+1}-1}N(t+1)=
\begin{cases}
\frac{1}{2}(q^{2}-q), &\hbox{if } k=0,\\
\frac{1}{2}q^{2k-1}(q-1)^{2}(q+1),& \hbox{if }  k \geq 1.\\
\end{cases}
\end{equation}
\end{lemma}

\begin{lemma}\label{lemma10}
Let $k$ and $a\in [1,q]$ be two positive integers,  and let $i$ be an integer.  Then we have 
\begin{IEEEeqnarray}{rcl}\label{wwwww}
    \sum\limits_{t= q^{k-i},q\nmid t}^{ a q^{k-i}-1}N(tq^{2i-1}+1)
    &=&\begin{cases}
        \frac{1}{2}(a^2-1)(q-1)^2q^{2k-3},\\ \quad \hbox{if }i\in [ -k+2,  k-1],\\
        \frac{1}{2} a(a-1)(q-1)q^{2k-2},\\\quad 
        \hbox{if }i=k \hbox{ or }-k+1,
        \end{cases}    \IEEEeqnarraynumspace
\end{IEEEeqnarray}
and 
\begin{IEEEeqnarray}{rcl}\label{wwwww2}
    &&\sum\limits_{t= q^{k-i},q\nmid t}^{ a q^{k-i}}N(tq^{2i}+1)\nonumber \\  
    &&\quad=\begin{cases}
        \frac{1}{2}(a^2-1)(q-1)^2q^{2k-2},\\ \quad \hbox{if }i\in [-k+1, k-1] \setminus
        \{0\},\\
        \frac{1}{2} a(a-1)(q-1)q^{2k-1},\\
        \quad 
        \hbox{if }i=k \hbox{ or }-k,\\       
       \frac{1}{2}(a^2-1)(q-1)^2q^{2k-2}
        +\frac{1}{2}(a-1)(q-1)q^{k-1}, \\ \quad \hbox{if }i=0.
        \end{cases}   
\end{IEEEeqnarray}
 \end{lemma}

\begin{assertion}\label{as1}
Suppose that $m\geq 4$ and  $q^{m-h}<\delta\leq q^{\lfloor(2m-1)/{3}\rfloor +1}$.   
\begin{itemize}
\item If   $m$ is  odd, then   \begin{equation}\label{as1-1}
\left| \left[\sum\limits_{\ell=h-k_{\delta}+1}^{h+k_{\delta}}\delta_{\ell}q^{\ell},\delta-1\right]\cap \mathcal{S}\right|=N(\mu(\delta)+1).
\end{equation}
\item 
If $m$ is even, then 
\begin{equation}\label{as1e1}
\left| \left[\sum\limits_{\ell=h-k_{\delta}}^{h+k_{\delta}}\delta_{\ell}q^{\ell},\delta-1\right]\cap (\mathcal{S}\cup \mathcal{H})\right|=N(\widetilde{\mu}(\delta)+1)
\end{equation} 
and 
\begin{equation}\label{ww1}
\left|\left[\sum\limits_{\ell=h}^{h+k_{\delta}}\delta_{\ell}q^{\ell}, \delta-1\right]\cap \mathcal{H}\right|= {\tau(\delta)}.
 \end{equation}
 \end{itemize}
\end{assertion}

\begin{assertion}\label{as2}
Suppose that $m\geq 4$ and   $q^{m-h}<\delta\leq q^{\lfloor (2m-1)/{3}\rfloor +1}$.  
\begin{itemize}
\item 
If $m$ is odd,  then 
\begin{multline}
\label{xx2}
\left|\left[q^{h+k_{\delta}},  \sum\limits_{\ell=h-k_{\delta}+1}^{h+k_{\delta}}\delta_{\ell}q^{\ell}\right)\cap \mathcal{S}\right |\\=\sum\limits_{i=-k_{\delta}+1}^{k_{\delta}} \sum\limits_{t\in \mathcal{T}_i(\delta)}N(tq^{2i-1}+1).
\end{multline}
\item 
If $m$ is  even, then 
\begin{multline}\label{aas31}
\left|\left[q^{h+k_{\delta}},\sum\limits_{\ell=h-k_{\delta}}^{h+k_{\delta}}\delta_{\ell}q^{\ell}\right)\cap 
(\mathcal{S}\cup \mathcal{H})\right|\\
=\sum\limits_{i=-k_{\delta}}^{k_{\delta}} \sum\limits_{t\in { \mathcal{T}}_{i}(\delta)}N(tq^{2i}+1) 
\end{multline}
and
 \begin{multline}   \label{aas32}
    \left|\left[q^{h+k_{\delta}}, \sum\limits_{\ell=h}^{h+k_{\delta}}\delta_{\ell}q^{\ell}\right)\cap \mathcal{H}\right|\\= N\left(\sum\limits_{\ell=h}^{h+k_{\delta}}\delta_{\ell}q^{\ell-h}\right)-N\left(q^{k_{\delta}}\right).
     \end{multline}
    \end{itemize}
\end{assertion}

\begin{assertion}\label{as3}
Suppose that $m\geq 4$ and  $ k\in \left[1, \lfloor (2m-1)/{3} \rfloor-h\right]$ is an integer. \begin{itemize}
\item 
If $m$ is  odd,  then 
\begin{equation}\label{10-0}
 \left|\mathcal{A}_{k}(i)\right|=
\begin{cases}
\frac{1}{2}q^{2k-1}\left(q-1\right)^2,  
\\
\quad \hbox{if }
 i=k \hbox{ or }-k+1, \\
\frac{1}{2}q^{2k-3}(q-1)^{3}(q+1),\\
\quad \hbox{if }i \in [-k+2, k-1]. 
\end{cases}
\end{equation}
\item 
If $m$ is  even,   then 
\begin{equation}\label{2-10-0}
\left|\mathcal{B}_{k}(i)\right|=\\
\begin{cases}
\frac{1}{2}q^{2k}(q-1)^2,\\ 
\quad \hbox{if } i=k\hbox{ or }-k,  \\
\frac{1}{2}q^{2k-2}(q-1)^{3}(q+1),\\\quad   \hbox{if }i\in [-k+1, k-1] \setminus
        \{0\}.
\end{cases}
\end{equation}
\end{itemize}
\end{assertion}

\begin{assertion}\label{as5}
Suppose that $m\geq 4$ is even and  $ k\in \left[0, \lfloor {(2m-1)}/{3} \rfloor-h\right]$ is an integer. Then 
    \begin{equation}\notag
\left|\mathcal{B}_k(0)\right|= 
\begin{cases}
\frac{1}{2}q(q-1),& \hbox{if }k=0,\\
\frac{1}{2}(q-1)^2(q^{2k}-q^{2k-2}+q^{k-1}),&  \hbox{if }k\geq 1.
\end{cases}
\end{equation}
\end{assertion}

\noindent\textbf{Proof of Theorem \ref{odda}.}
   If   $m$ is odd, we first aim to  
    show that 
    \begin{equation}\label{as6e1}
        |[1,\delta-1]\cap \mathcal{S}|=f(\delta)\quad \hbox{for }\delta \in \left[2, q^{\lfloor (2m-1)/{3}\rfloor +1}\right].
    \end{equation}
    
 If   $2\leq \delta\leq q^{h+1}$, recalling the definition of $\mathcal{S}$ and applying  Lemma \ref{lll},  we can  obtain \begin{equation}\label{colu}
     \left[1, \delta-1\right]\cap \mathcal{S}|=0.
 \end{equation}
 In particular, 
 \begin{equation}\label{ff1}
 \left|\left[1, q^{h+1}-1\right]\cap\mathcal{S}]\right| = 0.  
 \end{equation}

If   $q^{h+1}<\delta\leq  q^{\lfloor ({2m-1})/{3}\rfloor+1},$ since $m$ is odd, we first have 
\begin{itemize}
\item $1\leq   k_\delta\leq \lfloor ({2m-1})/{3}\rfloor-h;$
    \item   $q^{m-h}<\delta\leq q^{\lfloor ({2m-1})/{3}\rfloor+1}.$
\end{itemize}
By applying Theorem  \ref{th2}, we can conclude from  Assertion \ref{as3} that for any integer $ k\in \left[1,  \lfloor ({2m-1})/{3}\rfloor-h\right],$ 
\begin{IEEEeqnarray}{rCl}
 \left|\left[q^{h+k}, q^{h+k+1} \right)\cap \mathcal{S}\right| \nonumber 
&=&\sum\limits_{i=-k+1}^k\left |\mathcal{A}_k(i)\right| \nonumber \\
&= & (k-1)q^{2k-3}(q-1)^3(q+1) \nonumber \\
&&+ \>q^{2k-1}(q-1)^2. \nonumber
\end{IEEEeqnarray}
Therefore, we have 
\begin{IEEEeqnarray}{rCl}\notag
\left|\left[q^{h+1},q^{h+k_{\delta}}\right)\cap \mathcal{S}\right|  & =  &\sum\limits_{{k}=1}^{k_{\delta}-1}
\left|\left[q^{h+{k}},q^{h+{k}+1}\right)\cap \mathcal{S}\right|  \\ \nonumber
&=&\sum\limits_{{k}=1}^{k_{\delta}-1}({k}-1)q^{2{k}-3}(q-1)^3(q+1) \\
&&+\> \sum\limits_{{k}=1}^{k_{\delta}-1}q^{2{k}-1}(q-1)^2.\nonumber
\end{IEEEeqnarray}
It is easy to verify  that 
\begin{equation}\notag
\sum\limits_{{k}=1}^{k_{\delta}-1}q^{2{k}-1}(q-1)^2=\frac{(q^{2{k_{\delta}}-1}-q)(q-1)}{q+1}
\end{equation}
and 
\begin{multline*}
    \sum\limits_{k=1}^{k_{\delta}-1}({k}-1)q^{2{k}-3}(q-1)^3(q+1)\\=  \frac{\left[q-(k_{\delta}-1)q^{2k_{\delta}-3}+(k_{\delta}-2)q^{2k_{\delta}-1}\right](q-1)}{q+1}.
\end{multline*}
By adding these two sums together,  we have 
\begin{equation}\label{m6}
\left|\left[q^{h+1},q^{h+k_{\delta}}\right)\cap \mathcal{S}\right| =(k_{\delta}-1)q^{2k_{\delta}-3}(q-1)^2.
\end{equation}
We can also conclude from Assertions \ref{as1} and \ref{as2} that
\begin{IEEEeqnarray}{rCl}
\left|\left[q^{h+k_{\delta}}, \delta-1\right]\cap \mathcal{S}\right|&= &\left|\left[q^{h+k_{\delta}}, \sum\limits_{\ell=h-k_{\delta}+1}^{h+k_{\delta}}\delta_{\ell}q^{\ell}\right)\cap \mathcal{S}\right|   \nonumber \\ &&+\> \left|\left[\sum\limits_{\ell=h-k_{\delta}+1}^{h+k_{\delta}}\delta q^{\ell}, \delta-1\right]\cap \mathcal{S}\right| \nonumber\\ 
& =& \sum\limits_{i=-k_{\delta}+1}^{k_{\delta}}\sum\limits_{t\in \mathcal{T}_i(\delta)}N(tq^{2i-1}+1) \nonumber \\
&&+\> N(\mu(\delta)+1).  \label{mt0} 
\end{IEEEeqnarray}
It follows from (\ref{ff1}), (\ref{m6}) and (\ref{mt0})  that 
\begin{IEEEeqnarray}{rCl}
\left|\left[1,\delta-1\right]\cap \mathcal{S}\right| 
&=&\left|\left[1,q^{h+1}-1\right]\cap \mathcal{S}\right|
 + \left|\left[q^{h+1},q^{h+k_{\delta}}\right)\cap \mathcal{S}\right|\nonumber\\ &&+\>  \left|\left[q^{h+k_{\delta}}, \delta-1\right]\cap \mathcal{S}\right| \nonumber\\
&=&(k_{\delta}-1)q^{2k_{\delta}-3}(q-1)^2+N(\mu(\delta)+1)\nonumber \\
&&+\>
\sum\limits_{i=k_{\delta}+1}^{k_{\delta}}\sum\limits_{t\in\mathcal{T}_i(\delta)}N(tq^{2i-1}+1).\label{as54}
\end{IEEEeqnarray}

By the definition of the function $f(\delta)$, we now can conclude from (\ref{colu}) and (\ref{as54}) that  equation (\ref{as6e1}) holds.
Recalling the fact that an integer cannot be a coset leader if $q\mid a$, and the definition  of  $\mathcal{S}$,  it follows that  
\begin{equation}\label{eq64}
   \left| \mathcal{L}(1,\delta)\right|=N(\delta)-f(\delta)
\end{equation}for any integer $\delta \in \left[2, q^{\lfloor (2m-1)/{3}\rfloor +1}\right]$.
In addition, it is easy to verify that 
$\lfloor ({2m-1})/{3}\rfloor+1\leq m-\lfloor {m}/{3}\rfloor$ for odd $m$.  
This  implies that  $\delta-1<  q^{\lfloor ({2m-1})/{3}\rfloor+1}\leq q^{m-\lfloor {m}/{3}\rfloor}$. 
Then applying Theorem \ref{th1}, we have 
\begin{equation}\label{eq644}
    |C_a|=m \quad \hbox{for any integer }a\in [1, \delta-1].
\end{equation} 
Recalling equation (\ref{01}),  we can obtain     (\ref{th31}) from (\ref{eq64}) and (\ref{eq644}). 

If $m$ is even, our first goal is to  
to establish
\begin{equation}\label{v1}
\left|\left[1, \delta-1\right]\cap (\mathcal{S}\cup \mathcal{H})\right|= \widetilde{f}(\delta)
\end{equation}
and 
\begin{equation}\label{v2}
\left|\left[1, \delta-1\right]\cap \mathcal{H}\right|= g(\delta)
\end{equation}
for any integer  $\delta \in \left[2, q^{\lfloor (2m-1)/{3}\rfloor +1}\right]$.

If  $2\leq \delta\leq  q^{h}$,  we can apply Lemma \ref{lll} to conclude that
\begin{equation}\notag
\left|[1, \delta-1]\cap \mathcal{S}\right|=0.
\end{equation}
{In addition, by applying Corollary \ref{cos}, }we obtain  
\begin{equation}\notag
\left|[1, \delta-1]\cap \mathcal{H}\right|=0.
\end{equation}
 In particular,   we have  \begin{equation}\label{v5}
|[1, q^h-1]\cap (\mathcal{S}\cup \mathcal{H})| = 0
\end{equation}
and 
\begin{equation}\label{v6}
|[1, q^h-1]\cap \mathcal{H}| = 0.
\end{equation}

If $q^{h}< \delta\leq q^{h+1}$,  since $m$ is even, 
we first have 
\begin{itemize}
\item $k_\delta=0$;
    \item ${\mathcal{T}}_0(\delta)= \{t\in \mathbb{Z}: 1\leq t\leq \delta_h-1\};
    $\item $q^{m-h}<\delta\leq q^{\lfloor ({2m-1})/{3}\rfloor+1}$.
\end{itemize}
 Therefore,  by  substituting $k_{\delta}=0$ into  (\ref{as1e1}) and (\ref{aas31}), we obtain 
 \begin{equation}\notag
\left|\left[\delta_h q^{h},\delta-1 \right] \cap (\mathcal{S}\cup \mathcal{H})\right| 
 = N(\widetilde{\mu}(\delta)+1) 
  \end{equation}
 and  
\begin{IEEEeqnarray}{rCl}\notag
\left|\left[q^{h},\delta_h q^h \right) \cap (\mathcal{S}\cup \mathcal{H})\right|  =
 \sum\limits_{t\in \mathcal{T}_0(\delta)}N(t+1)
 = \frac{1}{2}(\delta_h-1)\delta_h. 
\end{IEEEeqnarray}
With (\ref{v5}),
it follows that 
\begin{IEEEeqnarray}{rcl}\label{shs1}
\left|\left[1,\delta-1 \right] \cap (\mathcal{S}\cup \mathcal{H})\right|  
& =& \frac{1}{2}(\delta_h-1)\delta_h +N(\widetilde{\mu}(\delta)+1).  \quad
\end{IEEEeqnarray}
Additionally, by substituting $k_{\delta}=0$ into (\ref{ww1}) and (\ref{aas32}), we obtain  
\begin{IEEEeqnarray}{rCl}\notag
\left|\left[\delta_h q^{h},\delta-1 \right] \cap \mathcal{H}\right| =
  \tau(\delta)\nonumber
  \end{IEEEeqnarray}
  and 
  \begin{equation}\notag
   \left|\left[q^{h},\delta_h q^h\right) \cap \mathcal{H}\right| =\delta_h-1.
\end{equation}
Combing  these two equalities with (\ref{v6}),
we obtain  
\begin{IEEEeqnarray}{rCl}\label{shs2}
\left|\left[1,\delta-1 \right] \cap  \mathcal{H}\right|  
& =& \delta_h-1+\tau(\delta).  
\end{IEEEeqnarray}
By now we have already  demonstrated that both (\ref{v1}) and (\ref{v2}) hold for $2\leq  \delta\leq q^{h+1}.$ 

Notice that $q^{h+1}-1=\sum\limits_{\ell=0}^h(q-1)q^{\ell}$. Therefore, for $\delta=q^{h+1}$,  we have 
  $\tau(\delta)=1$ and 
    $\widetilde{\mu}(\delta)=q-1$. 
Consequently, we can conclude from (\ref{shs1}) and (\ref{shs2}) that 
\begin{equation}\label{v7}
\left|\left[1, q^{h+1}-1\right]\cap (\mathcal{S}\cup \mathcal{H})\right|=\frac{1}{2}q(q-1)
\end{equation}
and 
\begin{equation}\label{v8}
\left|\left[1, q^{h+1}-1\right]\cap \mathcal{H}\right|= q-1.
\end{equation}

If  $q^{h+1}< \delta \leq q^{\lfloor ({2m-1})/{3}\rfloor+1 }$, since $m$ is even, we have 
\begin{itemize}
    \item  $1\leq k_\delta\leq \lfloor ({2m-1})/{3}\rfloor-h;$
    \item $q^{m-h}< \delta \leq q^{\lfloor ({2m-1})/{3}\rfloor+1 }$
\end{itemize}
By applying Theorem \ref{th2}, we can conclude from Assertions \ref{as3} and \ref{as5} to conclude that for  $k\in [1, \lfloor ({2m-1})/{3}\rfloor-h]$,
\begin{IEEEeqnarray}{rCl}
{\left|\left[q^{h+k},q^{h+k+1}\right)\cap (\mathcal{S}\cup \mathcal{H})\right|} &=&\sum\limits_{i=-k}^k  |\mathcal{B}_k(i) |\nonumber \\
&=& (k-\frac{1}{2})q^{2k-2}(q-1)^3(q+1)   \nonumber\\
&&+\> (q-1)^2(\frac{1}{2}q^{k-1}+q^{2k}).  \nonumber
\end{IEEEeqnarray}
It follows that
\begin{IEEEeqnarray}{rCl}
\IEEEeqnarraymulticol{3}{l}
{\left|\left[q^{h+1},q^{h+k_{\delta}}\right)\cap (\mathcal{S}\cup \mathcal{H})\right|}\nonumber \\
& =&\sum\limits_{k=1}^{k_{\delta}-1}
\left|\left[q^{h+k},q^{h+k+1}\right)\cap  (\mathcal{S}\cup \mathcal{H})\right|  \label{w3} \\ 
&=&(k_{\delta}-\frac{1}{2})q^{2k_{\delta}-2}(q-1)^2  +\frac{1}{2}(q-1)(q^{k_{\delta}-1}-q).  \nonumber
\end{IEEEeqnarray}
We can also conclude from Assertions \ref{as1} and \ref{as2} that 
\begin{IEEEeqnarray}{rCl}
\left|\left[q^{h+k_{\delta}}, \delta-1\right]\cap (\mathcal{S}\cup \mathcal{H})\right|
 &=&\sum\limits_{i=-k_{\delta}}^{k_{\delta}}\sum\limits_{t\in {\mathcal{T}}_i(\delta)}N(tq^{2i}+1) \nonumber \\
 &&+\>N(\widetilde{\mu}(\delta)+1).  \label{w2}
\end{IEEEeqnarray}
Combining (\ref{v7}), (\ref{w3}) and (\ref{w2}), we obtain 
\begin{IEEEeqnarray}{rCl}
\left|\left[1,  \delta-1\right]\cap \left(\mathcal{S}\cup \mathcal{H}\right)\right| \nonumber &=& \left|\left[1,q^{h+1} \right)\cap (\mathcal{S}\cup \mathcal{H})\right| \nonumber\\
&&+\> \left|\left[q^{h+1}, q^{h+k_{\delta}}\right)\cap (\mathcal{S}\cup \mathcal{H})\right| \nonumber \\
&&+\> \left|\left[q^{h+k_{\delta}}, \delta-1\right]\cap (\mathcal{S}\cup \mathcal{H})\right|  \nonumber \\
    &= &(k_{\delta}-\frac{1}{2})q^{2k_{\delta}-2}(q-1)^2 \nonumber \\
    && +\>\frac{1}{2}q^{k_{\delta}-1}(q-1) + N(\widetilde{\mu}(\delta)+1) \nonumber \\
    &&+\>  \sum\limits_{i=-k_{\delta}}^{k_{\delta}}\sum\limits_{t\in {\mathcal{T}}_i(\delta)}N(tq^{2i}+1)  \nonumber
\end{IEEEeqnarray}
for  $q^{h+1}< \delta \leq q^{\lfloor ({2m-1})/{3}\rfloor+1 }$. 

In addition, 
by applying Corollary \ref{h}, we have 
\begin{equation}\notag
\begin{split}
   \left |\left[q^{h+1}, q^{h+k_{\delta}}\right)\cap \mathcal{H}   \right|&=\sum\limits_{k=1}^{k_{\delta}-1}|[q^{h+k}, q^{h+{k}+1})\cap \mathcal{H}|\\
     &= (q^{k_{\delta}-1}-1)(q-1).
\end{split}
\end{equation}
With (\ref{v8}), it follows that 
\begin{equation}\label{773}
\left|\left[1,q^{h+k_{\delta}}\right)\cap \mathcal{H}\right|  =  q^{k_{\delta}-1}(q-1).
\end{equation}
Noting that $N(q^{k_{\delta}})=q^{k_{\delta}-1}(q-1)$ for 
$k_{\delta}\geq 1$, we can combine (\ref{ww1}), (\ref{aas32}) and (\ref{773})  to obtain 
\begin{IEEEeqnarray} {rcl}
  \left |\left[1,\delta\!-\!1\right]\cap \mathcal{H}\right|
      &=&   |[1,q^{h+k_{\delta}})\cap \mathcal{H}|+  \left|\left[q^{h+k_{\delta}}, \sum\limits_{\ell=h}^{h+k_{\delta}}\delta_{\ell}q^{\ell}\right)\cap \mathcal{H}\right|\nonumber \\    &&+\>\left|\left[\sum\limits_{\ell=h}^{h+k_{\delta}}\delta_{\ell}q^{\ell}, \delta-1\right]\cap \mathcal{H}\right|\nonumber\\
     & =& N\left(\sum\limits_{\ell=h}^{h+k_{\delta}}\delta_{\ell}q^{\ell-h}\right)+\tau(\delta). \nonumber
\end{IEEEeqnarray}
We now have demonstrated  that (\ref{v1}) and (\ref{v2}) also  hold for $q^{h+1}< \delta \leq q^{\lfloor ({2m-1})/{3}\rfloor +1}$.

Next, we establish equation (\ref{th311}). For simplicity, we define the following two sets for positive integers $b$ and $\delta$:
\begin{equation*}
    \mathcal{L}_m(b,\delta)=\left\{a\in  \mathcal{L}(b,\delta): |C_a|=m\right\},
    \end{equation*}
    \begin{equation*}
    \mathcal{L}_{\frac{m}{2}}(b,\delta)=\left\{a\in  \mathcal{L}(b,\delta): |C_a|=m/2\right\}.
    \end{equation*}
Recall  the definitions of $\mathcal{S}$ and $\mathcal{H}$.  We can conclude from (\ref{v1}) and (\ref{v2})  that  \begin{equation}\label{zyb}
|\mathcal{L}_m(1,\delta)|=N(\delta)-\widetilde{f}(\delta)\quad \hbox{and}\quad |\mathcal{L}_{\frac{m}{2}}(1,\delta)|=g(\delta).
\end{equation}
Note that $\delta \leq q^{\lfloor ({2m-1})/{3}\rfloor +1}\leq q^{m-\lfloor  {m}/{3} \rfloor}$ for even $m$. Thus, by
Theorem \ref{th1},  we have $C_a=m$ or $|C_a|={m}/{2} $ for any integer $a\in [1, \delta-1]$.  
 Consequently, we conclude that the equality in   (\ref{th311}) holds.  This completes the proof of Theorem \ref{odda}. \qed

As examples of Theorem \ref{odda}, we present the following two corollaries. It can be verified that these corollaries are equivalent to Theorem 17 and Theorem 18 in \cite{liu2017}, respectively.

\begin{Corollary}
 Let $m\geq 4$ be an even integer and  $\delta\in [2,q^{h+1}]$ be an integer. Then 
\begin{equation}\notag
    \mathrm{dim}(\mathcal{C}_{(q,m,\delta)})=
    \begin{cases}
        n-m N(\delta),   \hbox{ if }\delta\leq q^{h},\\
        n-m\left(N(\delta)-\frac{1}{2}\delta_h^2\right),\\
        \quad \hbox{if }\delta>q^{h} \hbox{ and }\delta_h\leq \sum\limits_{\ell=0}^{h-1}\delta_{\ell}q^{\ell},\\
      n-m\left[N(\delta)-\frac{1}{2}(\delta_h-1)^2-\delta_0  \right],  \\
      \quad \hbox{if }\delta>q^h \hbox{ and }\delta_h>\sum\limits_{\ell=0}^{h-1}\delta_{\ell}q^{\ell}.
    \end{cases}
\end{equation}
\end{Corollary} 
\begin{proof}
We consider the following cases:

 \textbf{Case 1. } If $ \delta\leq q^h$, we immediately have $\widetilde{f}(\delta)=0$ and  $g(\delta)=0$ by definition.  

\textbf{Case 2. } If  $\delta>q^h,$  we can observe that $\delta_h>0$ and  $s_{\delta}=k_{\delta}=0$. Then we distinguish the following two subcases:

  \textit{Subcase 1.} If $\delta_h\leq \sum\limits_{\ell=0}^{h-1}\delta_{\ell}q^{\ell}$, then  $\tau(\delta)=1$ and   $\widetilde{\mu}(\delta)=\delta_h,$ which implies that 
    \begin{equation}\notag
\widetilde{f}(\delta)=\frac{1}{2}(\delta_h-1)\delta_h+\delta_h  \quad \hbox{and}\quad 
g(\delta)=\delta_h.
 \end{equation}
 
 \textit{Subcase 2.} If $\delta_h>\sum\limits_{\ell=0}^{h-1}\delta_{\ell}q^{\ell}$, then $\tau(\delta)=0$ and $\widetilde{\mu}(\delta)=\sum\limits_{\ell=0}^{h-1}\delta_{\ell}q^{\ell}=\delta_0$, which implies that 
 \begin{equation}\notag
   \widetilde{f}(\delta)=\frac{1}{2}(\delta_h-1)\delta_h+\delta_0\quad \hbox{and}\quad g(\delta)=  \delta_h-1.
 \end{equation}

By substituting the values of $\widetilde{f}(\delta)$ and $g(\delta)$ into equation (\ref{th311}) for each corresponding case, we obtain the desired equality.
\end{proof}
    
\begin{Corollary}
 Let $m\geq 5$ be an odd integer and  $\delta\in [2,q^{h+2}]$ be an integer. Then 
\begin{multline}\notag
\mathrm{dim}(\mathcal{C}_{(q,m,\delta)})=\\ 
    \begin{cases}
        n-mN(\delta), \hbox{ if 
 }\delta\leq q^{h+1},\\
        n-m\left[N(\delta)-(q-1)\delta_{h+1}^2-(\delta_h-1)\delta_{h+1}-\delta_0\right], \\ \quad \hbox{if } \delta>q^{h+1}, \delta_h>0 \hbox{ and } \delta_{h+1}>\sum\limits_{\ell=0}^{h-1}\delta_{\ell}q^{\ell},\\
        n-m\left[N(\delta)-(q-1)\delta_{h+1}^2-\delta_h\delta_{h+1}\right], \\ \quad\hbox{if } \delta>q^{h+1}, \delta_h>0 \hbox{ and } \delta_{h+1}\leq \sum\limits_{\ell=0}^{h-1}\delta_{\ell}q^{\ell},\\
n-m\left[N(\delta)-(q-1)(\delta_{h+1}^2-\delta_{h+1}+\delta_1)-\delta_0\right],\\ \quad \hbox{if }\delta>q^{h+1}, \delta_h=0\hbox{ and }\delta_{h+1}q> \sum\limits_{\ell=0}^{h-1}\delta_{\ell}q^{\ell},\\
        n-m\left[N(\delta)-(q-1)\delta_{h+1}^2\right],\\ \quad \hbox{if }\delta>q^{h+1}, \delta_h=0 \hbox{ and }\delta_{h+1}q\leq \sum\limits_{\ell=0}^{h-1}\delta_{\ell}q^{\ell}.
    \end{cases}
\end{multline}
 \end{Corollary}
\begin{proof}
    We consider the following cases:
    
    \textbf{Case 1.} If $ \delta\leq q^{h+1}$, we can directly have $f(\delta)=0$  by definition.

    \textbf{Case 2.} If $\delta> q^{h+1}$, we have $k_{\delta}=1$, then we distinguish the following subcases:

\textit{Subcase 1.} 
 If $\delta_h>0$, then $s_{\delta}=0$, which leads to 
 \begin{itemize}
     \item $\mathcal{T}_0(\delta)=\left[q, \delta_{h+1}q+ \delta_h-1\right]\cap \{t\in \mathbb{Z}:q\nmid t\}$; 
     \item 
     $\mathcal{T}_1(\delta)=\left[1, \delta_{h+1}\right]$;
      \item  $\mu(\delta)=\min\left\{\delta_{h+1},\sum\limits_{\ell=0}^{h-1}\delta_{\ell}q^{\ell}\right\}$.
      \end{itemize}
      Noticing that $\mu(\delta)=\sum\limits_{\ell=0}^{h-1}\delta_{\ell}q^{\ell}=\delta_0$ if $\delta_{h+1}>\sum\limits_{\ell=0}^{h-1}\delta_{\ell}q^{\ell}$, we can 
      substitute these terms into equation (\ref{ff}) to  obtain 
\begin{equation}\notag  f(\delta)=\begin{cases}
    (q-1)\delta_{h+1}^2+(\delta_h-1)\delta_{h+1}+\delta_0, \\ \quad \quad\hbox{if }\delta_h>0 \hbox{ and } \delta_{h+1}>\sum\limits_{\ell=0}^{h-1}\delta_{\ell}q^{\ell},\\
     (q-1)\delta_{h+1}^2+\delta_h\delta_{h+1},  \\ \quad \quad \hbox{if }\delta_h>0 \hbox{ and } \delta_{h+1}\leq \sum\limits_{\ell=0}^{h-1}\delta_{\ell}q^{\ell}.
\end{cases}
\end{equation} 

 \textit{Subcase 2.} If $\delta_h=0$, then $s_{\delta}=1$, which leads to 
\begin{itemize}
     \item $\mathcal{T}_0(\delta)=\left[q, \delta_{h+1}q-1\right]\cap \{t\in \mathbb{Z}:q\nmid t\}$;
    \item  $\mathcal{T}_1(\delta)=[1,\delta_{h+1}-1]$;
    \item  $\mu(\delta)=\min\left\{ \delta_{h+1}q, \sum\limits_{\ell=0}^{h-1}\delta_{\ell}q^{\ell} \right\}$.
\end{itemize}
Noting that  $\mu(\delta)=\sum\limits_{\ell=0}^{h-1}\delta_{\ell}q^{\ell}=\delta_1q+\delta_0$ if $\delta_{h+1}q>\sum\limits_{\ell=0}^{h-1}\delta_{\ell}q^{\ell}$, 
we can substitute  these  terms into equation (\ref{ff}) to derive 
\begin{equation}\notag
    f(\delta)= \begin{cases}
     (q-1)\delta_{h+1}^2, \hbox{ if } \delta_h=0 \hbox{ and }\delta_{h+1}q\leq \sum\limits_{\ell=0}^{h-1}\delta_{\ell}q^{\ell},\\
        (q-1)(\delta_{h+1}^2-\delta_{h+1}+\delta_1)+\delta_0,\\ \quad \quad \hbox{if }\delta_h=0 \hbox{ and }\delta_{h+1}q>\sum\limits_{\ell=0}^{h-1}\delta_{\ell}q^{\ell}. 
    \end{cases}
\end{equation}
  
Finally, substituting the value of $f(\delta)$ into equation (\ref{th31}) for each corresponding case yields the result. 
\end{proof}
\section{The Dimension of \texorpdfstring{$\mathcal{C}_{(q,m,\delta,b)} $ }  \ for   \texorpdfstring{ $\delta+b\leq q^{\lfloor ({2m-1})/{3} \rfloor+1}+1$}\ }  
Based on the proof of Theorem \ref{odda}, we can easily generalize the results in Theorem \ref{odda} to primitive BCH codes that are not necessarily narrow-sense as follows.
\begin{theorem}\label{ext}
 Let $m\geq 4$, and let  $b$ and $\delta$ be two positive integers such that $b\geq 2$ and   $ b(q-1)+1\leq \delta \leq q^{\lfloor ({2m-1})/{3} \rfloor+1}+1-b$. 
 \begin{itemize}
    \item If  $m$ is odd, then 
\begin{equation*}
\mathrm{dim}(\mathcal{C}_{(q,m,\delta, b)})=
n-m\left[N(b+\delta-1)-f(b+\delta-1)\right]. 
\end{equation*}
\item If $m$ is even, then 
\begin{equation*}
\begin{split}
\mathrm{dim}(\mathcal{C}_{(q,m,\delta, b)}) 
=  & n - m\left[N(b+\delta-1)-\widetilde{f}(b+\delta\!-\!1)\right]\\
& -\frac{m}{2}g(b+\delta\!-\!1).
\end{split}
\end{equation*}  
\end{itemize} 
\end{theorem}
\begin{proof}
Since $\delta\geq  b(q-1)+1,$ for each integer $x\in [1,b-1]$,   there exists $i$ such that $b\leq xq^i\leq bq-1\leq b+\delta-2$.  Therefore, 
\begin{equation*}
    C_{x}\subseteq \bigcup\limits_{a=b}^{b+\delta-2}C_a\quad \hbox{for all } x\in [1,b-1]. 
\end{equation*}
It follows that $$\bigcup\limits_{a=1}^{b+\delta-2}C_a=\bigcup\limits_{a=b}^{b+\delta-2}C_a,$$ which is equivalent to 
$  \mathcal{C}_{(q,m,\delta,b)}=\mathcal{C}_{(q,m,b+\delta-1)}.
$
Noticing that $b+\delta-1 \leq q^{\lfloor ({2m-1})/{3} \rfloor+1}$, we can apply Theorem 3 to obtain the result.
\end{proof}

\section{The Bose distance of \texorpdfstring{$\mathcal{C}_{(q,m,\delta,b)}$ }  \ for   \texorpdfstring{ $ \delta+b\leq q^{\lfloor ({2m-1})/{3} \rfloor+1}+1$}\ } 
 When
$\delta=q^{\lfloor (2m-1)/3\rfloor+1}$, the Bose distance was  given
by \cite[Theorem 13]{DDZ2015}. 
In  this section, we determine the dimension of $C_{(q,m,\delta)}$ for $\delta \in [2, q^{\lfloor ({2m-1})/{3}\rfloor+1}).$

We can first apply Lemma \ref{lll} \cite[Theorem 2.3]{YF2000} to determine the Bose distance of $\mathcal{C}_{(q,m,\delta)}$ for $\delta\in [2,q^{m-h})$ as follow. 
\begin{theorem} \label{bosth0}
Let $\delta\in [2, q^{m-h})$ be an integer. Then 
\begin{equation}\notag d_B(\mathcal{C}_{(q,m,\delta)})= \begin{cases}
\delta, &\hbox{ if }q\nmid \delta, \\ 
\delta+1, & \hbox{ if } q\mid \delta.\\
  \end{cases}
\end{equation}
\end{theorem}
\begin{proof}
Note that $m-h=h$ when $m$ is odd, and $m-h=h+1$ when $m$ is even. We can apply Lemma \ref{lll} to conclude that $\delta$ is the smallest coset leader in $[\delta, n]$ if $q\nmid \delta$,  and  $\delta+1$ is the smallest coset leader in $[\delta, n]$ if $q\mid \delta.$    Then the result follows. 
\end{proof}

Next, we determine the 
 the Bose distance of $\mathcal{C}_{(q,m,\delta)}$ for $\delta \in [q^{m-h}, q^{\lfloor ({2m-1})/{3}\rfloor+1})$ and  $m\geq 4$ by giving  the following Theorem \ref{bosth} and Theorem \ref{odd2}. 

\begin{theorem}\label{bosth}
Suppose that $m\geq 5$ is an odd integer and  $\delta\in \left[q^{h+1}, q^{\lfloor({2m-1})/{3}\rfloor +1}\right)$ is  an integer. Let $j_{\delta}$ be the integer such that $q^{h+j_{\delta}}\leq \delta<q^{h+j_{\delta}+1}$. Let  $\sum\limits_{\ell=0}^{h+j_{\delta}}\delta_{\ell}q^{\ell}$ denote the $q$-adic expansion of $\delta$, and let $r_{\delta}$ be  the smallest integer in $[-j_{\delta}+1,j_{\delta}]$ such that $\delta_{h+r_{\delta}}>0$.  Define $\hat{\delta}= \sum\limits_{\ell=h+r_{\delta}}^{h+j_{\delta}}\delta_{\ell}q^{\ell}+ \sum\limits_{\ell=h-r_{\delta}+1}^{h+j_{\delta}}\delta_{\ell}q^{\ell-(h-r_{\delta}+1)}.$   
\begin{itemize}
\item  If  $q\nmid \delta$, then 
\begin{multline}\label{t1}
   d_B(\mathcal{C}_{(q,m,\delta)})=\\ \begin{cases}
\delta, & \hbox{if }  \sum\limits_{\ell=0}^{h-j_{\delta}}\delta_{\ell}q^{\ell}> \sum\limits_{\ell=h-r_{\delta}+1}^{h+j_{\delta}}\delta_{\ell}q^{\ell-(h-r_{\delta}+1)},\\
\hat{\delta}+1, &\hbox{if }\sum\limits_{\ell=0}^{h-j_{\delta}}\delta_{\ell}q^{\ell}\leq  \sum\limits_{\ell=h-r_{\delta}+1}^{h+j_{\delta}}\delta_{\ell}q^{\ell-(h-r_{\delta}+1)}\\
& \hbox{and } \delta_{h-r_{\delta}+1}\neq q-1, \\
\hat{\delta}+2, &\hbox{if }\sum\limits_{\ell=0}^{h-j_{\delta}}\delta_{\ell}q^{\ell}\leq  \sum\limits_{\ell=h-r_{\delta}+1}^{h+j_{\delta}}\delta_{\ell}q^{\ell-(h-r_{\delta}+1)} \\
& \hbox{ and } \delta_{h-r_{\delta}+1}=  q-1. \\     \end{cases}
\end{multline}
\item If  $q\mid \delta$, then   
\begin{multline}\label{t2}
   d_B(\mathcal{C}_{(q,m,\delta)})= \\\begin{cases}
\delta+1, & \hbox{if } \sum\limits_{\ell=0}^{h-j_{\delta}}\delta_{\ell}q^{\ell}\geq \sum\limits_{\ell=h-r_{\delta}+1}^{h+j_{\delta}}\delta_{\ell}q^{\ell-(h-r_{\delta}+1)},\\
\hat{\delta}+1, &\hbox{if }  \sum\limits_{\ell=0}^{h-j_{\delta}}\delta_{\ell}q^{\ell}<\sum\limits_{\ell=h-r_{\delta}+1}^{h+j_{\delta}}\delta_{\ell}q^{\ell-(h-r_{\delta}+1)} \\
&\hbox{and }\delta_{h-r_{\delta}+1}\neq q-1,\\
\hat{\delta}+2, &\hbox{if } \sum\limits_{\ell=0}^{h-j_{\delta}}\delta_{\ell}q^{\ell}<\sum\limits_{\ell=h-r_{\delta}+1}^{h+j_{\delta}}\delta_{\ell}q^{\ell-(h-r_{\delta}+1)} \\
& \hbox{and } \delta_{h-r_{\delta}+1}= q-1. 
\end{cases}
\end{multline}
\end{itemize}
\end{theorem}
\begin{proof}

 We will show that equation (\ref{t1})  holds if $q\nmid \delta$ through the following cases. 

\textbf{Case 1.}
Suppose that  $\sum\limits_{\ell=0}^{h-j_{\delta}}\delta_{\ell}q^{\ell}> \sum\limits_{\ell=h-r_{\delta}+1}^{h+j_{\delta}}\delta_{\ell}q^{\ell-(h-r_{\delta}+1)}$. 
By Remark \ref{rm6}, we first have  $\delta\not \in \mathcal{A}_{j_{\delta}}(r_{\delta})$.  Futhermore, recalling the definition of $r_{\delta} $ and applying 
Remark \ref{cor2},  we can also conclude that $\delta\not \in \mathcal{A}_{j_{\delta}}(i)$ for any $i\neq r_{\delta}$. By applying Theorem \ref{th2}, it follows that  $\delta \not \in \mathcal{S}$, which  implies that 
$\delta$
is a coset leader.
 Recalling equaiton (\ref{Bose1}),   
it follows that $ d_B(\mathcal{C}_{(q,m,\delta)})=\delta$. 

\textbf{Case 2.}
Suppose that  $\sum\limits_{\ell=0}^{h-j_{\delta}}\delta_{\ell}q^{\ell}\leq  \sum\limits_{\ell=h-r_{\delta}+1}^{h+j_{\delta}}\delta_{\ell}q^{\ell-(h-r_{\delta}+1)}$.  We first aim to show that any integer $\delta^{'}\in \left[\delta, \hat{\delta}\right]$ is not a coset leader.    Noticing  $\sum\limits_{\ell=h-r_{\delta}+1}^{h+j_{\delta}}\delta_{\ell}q^{\ell-(h-r_{\delta}+1)}<q^{j_{\delta}+r_{\delta}}$,   we have 
$\sum\limits_{\ell=0}^{h-j_{\delta}}\delta_{\ell}q^{\ell}<q^{j_{\delta}+r_{\delta}}$.  This implies that 
\begin{equation}\notag
    \delta_{\ell}=0 \quad \hbox{for }\ell=j_{\delta}+r_{\delta},\ldots, h-j_{\delta}.
\end{equation}
Recalling the definition of $r_{\delta}$, it follows that 
\begin{equation}\notag
    V(\delta)=(\mathbf{0}_{h-j_{\delta}}, \delta_{h+j_{\delta}},\ldots, \delta_{h+r_{\delta}}, \mathbf{0}_{h-j_{\delta}}, \delta_{j_{\delta}+r_{\delta}-1},\ldots,
    \delta_{0}).
\end{equation}
It is  clear that 
\begin{multline}\notag
    V(\hat{\delta})=(\mathbf{0}_{h-j_{\delta}}, \delta_{h+j_{\delta}},\ldots, \delta_{h+r_{\delta}}, \\\mathbf{0}_{h-j_{\delta}}, \delta_{h+j_{\delta}},\ldots,
    \delta_{h-r_{\delta}+1}).
\end{multline}
Consequently, for any integer $\delta^{'}\in \left[\delta, \hat{\delta}\right]$ with $q$-adic expansion $\sum\limits_{\ell=0}^{h+j_{\delta}}\delta^{'}_{\ell}q^{\ell}$,  we have 
\begin{equation}\label{dcj1}
\begin{split}
V(\delta^{'})&=(\mathbf{0}_{h-j_{\delta}}, \delta^{'}_{h+j_{\delta}}, \ldots.
    \delta_{0}^{'})\\
    &=(\mathbf{0}_{h-j_{\delta}}, \delta_{h+j_{\delta}}^{},\ldots, \delta_{h+r_{\delta}}, \mathbf{0}_{h-j_{\delta}}, \delta_{j_{\delta}+r_{\delta}-1}^{'},\ldots.
    \delta_{0}^{'})
    \end{split}
\end{equation}
and 
\begin{equation}\label{dcj2}
\begin{split}
 (\delta^{'}_{j_{\delta}+r_{\delta}-1}, \ldots, \delta^{'}_0)&\leq (\delta_{h+j_{\delta}}, \ldots,\delta_{h-r_{\delta}+1})\\&=(\delta^{'}_{h+j_{\delta}}, \ldots,\delta^{'}_{h-r_{\delta}+1}).
 \end{split}
\end{equation}
If $\delta_{0}^{'}=0$, then $q\mid \delta^{'}$, indicating that  $\delta^{'}$ is not a coset leader. 
If $\delta_0^{'}\neq 0$, then we can conclude from (\ref{dcj1}) and (\ref{dcj2}) that $\delta^{'}\in \mathcal{A}_{j_{\delta}}(r_{\delta})$. Thus, by applying Theorem \ref{th2}, we can still can conclude that $\delta^{'}$ is not a coset leader. 

In addition,  we can use a similar argument in Case 1 to conclude that 
\begin{itemize}
    \item If $\delta_{h-r_{\delta}+1}\neq q-1$, then $\hat{\delta}+1$ is a coset leader.
    \item If $\delta_{h-r_{\delta}+1}= q-1$, then  $\hat{\delta}+2$ is a coset leader.
\end{itemize}   
Furthermore, since  $q\mid \hat{\delta}+1$ when $\delta_{h-r_{\delta}+1}=q-1$, we can also conclude that $\hat{\delta}+1$ is not a coset leader in this case.  Consequently, we  conclude that   $d_B(\mathcal{C}_{(q,m,\delta)})=\hat{\delta}+1$ if $\delta_{h-r_{\delta}+1}\neq q-1$, and 
$d_B(\mathcal{C}_{(q,m,\delta)})=\hat{\delta}+2$ if $\delta_{h-r_{\delta}+1}= q-1.$ We  now have demonstrated that equation (\ref{t1}) holds if $q\nmid \delta$. 

If $q \mid \delta$, then $\delta$ is not a coset leader. It follows that $d_B(\mathcal{C}_{(q,m,\delta)})=d_B(\mathcal{C}_{(q,m,\delta+1)})$. Notably $q\mid \delta$ implies that $q\nmid \delta+1$. Therefore, we can obtain (\ref{t2}) by substituting $\delta+1$ into  equation (\ref{t1}).  \end{proof}

\begin{theorem}\label{odd2}
Suppose that   $m\geq 4$ is an even integer,  
and  $\delta\in [q^h, q^{\lfloor({2m-1})/{3}\rfloor +1})$ is  an  integer.   
Let $j_{\delta}$ be the integer such that $q^{h+j_{\delta}}\leq \delta<q^{h+j_{\delta}+1}$. Let  $\sum\limits_{\ell=0}^{h+j_{\delta}}\delta_{\ell}q^{\ell}$ denote the $q$-adic expansion of $\delta$, and let $r_{\delta}$ be the smallest integer in $[-j_{\delta},j_{\delta}]$ such that $\delta_{h+r_{\delta}}>0$. Define $\hat{\delta}= \sum\limits_{\ell=h+r_{\delta}}^{h+j_{\delta}}\delta_{\ell}q^{\ell}+ \sum\limits_{\ell=h-r_{\delta}}^{h+j_{\delta}}\delta_{\ell}q^{\ell-(h-r_{\delta})}.$  
\begin{itemize}
    \item 
If $r_{\delta}\neq 0$ and  $q\nmid \delta$, then   
\begin{multline}\label{sh3}
   d_B(\mathcal{C}_{(q,m,\delta)})= \\ \begin{cases}
\delta, & \hbox{if } \sum\limits_{\ell=0}^{h-j_{\delta}-1}\delta_{\ell}q^{\ell}> \sum\limits_{\ell=h-r_{\delta}}^{h+j_{\delta}}\delta_{\ell}q^{\ell-(h-r_{\delta})}, \\
\hat{\delta}+1, &\hbox{if }\sum\limits_{\ell=0}^{h-j_{\delta}-1}\delta_{\ell}q^{\ell}\leq  \sum\limits_{\ell=h-r_{\delta}}^{h+j_{\delta}}\delta_{\ell}q^{\ell-(h-r_{\delta})}\\
& \hbox{   and } \delta_{h-r_{\delta}}\neq q-1, \\
\hat{\delta}+2, &\hbox{if }\sum\limits_{\ell=0}^{h-j_{\delta}-1}\delta_{\ell}q^{\ell}\leq  \sum\limits_{\ell=h-r_{\delta}}^{h+j_{\delta}}\delta_{\ell}q^{\ell-(h-r_{\delta})} \\
& \hbox{and } \delta_{h-r_{\delta}}=  q-1.\\     \end{cases}
\end{multline}
\item If $r_{\delta}\neq 0$ and  $q\mid \delta$, then   
\begin{multline}\label{sh4}
   d_B(\mathcal{C}_{(q,m,\delta)}) = \\ \begin{cases}
\delta+1, & \hbox{if } \sum\limits_{\ell=0}^{h-j_{\delta}-1}\delta_{\ell}q^{\ell}\geq  \sum\limits_{\ell=h-r_{\delta}}^{h+j_{\delta}}\delta_{\ell}q^{\ell-(h-r_{\delta})}, \\
\hat{\delta}+1, &\hbox{if }\sum\limits_{\ell=0}^{h-j_{\delta}-1}\delta_{\ell}q^{\ell}<\sum\limits_{\ell=h-r_{\delta}}^{h+j_{\delta}}\delta_{\ell}q^{\ell-(h-r_{\delta})}\\
& \hbox{and } \delta_{h-r_{\delta}}\neq q-1,\\
\hat{\delta}+2, &\hbox{if }\sum\limits_{\ell=0}^{h-j_{\delta}-1}\delta_{\ell}q^{\ell}<  \sum\limits_{\ell=h-r_{\delta}}^{h+j_{\delta}}\delta_{\ell}q^{\ell-(h-r_{\delta})}\\
& \hbox{and } \delta_{h-r_{\delta}}=  q-1.\\     \end{cases}
\end{multline}
\item If   $r_{\delta}=0$,   then 
\begin{multline}\label{sh2}
   d_B(\mathcal{C}_{(q,m,\delta)})= \\ \begin{cases}
\delta, &\hbox{ if }\sum\limits_{\ell=0}^{h-j_{\delta}-1}\delta_{\ell}q^{\ell}> \sum\limits_{\ell=h}^{h+j_{\delta}}\delta_{\ell}q^{\ell-h}\\ & \hbox{ and } q\nmid \delta,\\ 
\delta+1, &\hbox{ if }\sum\limits_{\ell=0}^{h-j_{\delta}-1}\delta_{\ell}q^{\ell}>\sum\limits_{\ell=h}^{h+j_{\delta}}\delta_{\ell}q^{\ell-h}\\ & \hbox{ and } q\mid \delta,\\ 
\hat{\delta} , &\hbox{ if }\sum\limits_{\ell=0}^{h-j_{\delta}-1}\delta_{\ell}q^{\ell-h}\leq \sum\limits_{\ell=h}^{h+j_{\delta}}\delta_{\ell}q^{\ell-h}. \\ 
  \end{cases}
  \end{multline}
  \end{itemize}
\end{theorem}
\begin{proof}
First, by an argument similar to that in the proof of Theorem \ref{bosth}, we conclude that \eqref{sh3} holds if 
 $r_{\delta}\neq 0$ and $q\nmid \delta$, and that \eqref{sh4} holds 
if $r_{\delta}\neq 0$ and $q\mid \delta$.
Next, we will show that \eqref{sh2} holds 
if   $r_{\delta}=0$  by considering the following cases.

\textbf{Case 1.} Suppose that  $\sum\limits_{\ell=0}^{h-j_{\delta}-1}\delta_{\ell}q^{\ell}>\sum\limits_{\ell=h}^{h+j_{\delta}}\delta_{\ell}q^{\ell-h}$ and $q\nmid \delta$. 
Then by applying Remarks \ \ref{rm6} and \ref{cor2}, we see that 
$\delta\not \in \mathcal{B}_{j_{\delta}}(i) $ for any $i\in [-j_{\delta}, j_{\delta}].$ Hence $a\not \in \mathcal{S}\cup \mathcal{H}$. Therefore, we can conclude $\delta$ is a coset leader. 
 Consequently, we have $d_B(\mathcal{C}_{(q,m,\delta)})=\delta$ in this case.

\textbf{Case 2.}  Suppose that  $\sum\limits_{\ell=0}^{h-j_{\delta}-1}\delta_{\ell}q^{\ell}> \sum\limits_{\ell=h}^{h+j_{\delta}}\delta_{\ell}q^{\ell-h}$   and  $q\mid \delta$.  Then  the assumption that $ q\mid\delta$ implies that $\delta$  is not a coset leader. Furthermore, by an argument similar to that in   Case 1, we conclude that 
$\delta+1$ is a coset leader. 
 Therefore, $d_B(\mathcal{C}_{(q,m,\delta)})=\delta+1$. 

\textbf{Case 3.} Suppose that $\sum\limits_{\ell=0}^{h-j_{\delta}-1}\delta_{\ell}q^{\ell-h}\leq \sum\limits_{\ell=h}^{h+j_{\delta}}\delta_{\ell}q^{\ell-h}.$ Then we distinguish the following two subcases:

\textit{Subcase 1.} If $\sum\limits_{\ell=0}^{h-j_{\delta}-1}\delta_{\ell}q^{\ell}=\sum\limits_{\ell=h}^{h+j_{\delta}}\delta_{\ell}q^{\ell-h},$ then we have 
\begin{equation}\notag
    \delta_{\ell}=0 \quad \hbox{for }\ell= j_{\delta}+1,\ldots, h-j_{\delta}-1 
\end{equation}
and 
\begin{equation}\notag
    \delta_{\ell}=\delta_{\ell+h} \quad \hbox{for }\ell=0,\ldots, j_{\delta}.
\end{equation}
It follows that $\hat{\delta}=\delta$,  and
$V(\delta)$ has the form 
\begin{equation}\notag
    (\mathbf{0}_{h-j_{\delta}-1},\delta_{j_{\delta}},
    \ldots,\delta_0 , \mathbf{0}_{h-j_{\delta}-1},\delta_{j_{\delta}},\ldots,\delta_0 )
\end{equation}
with $\delta_0>0$ and $\delta_{j_{\delta}}>0$. By applying Corollary \ref{corr}, we can conclude that $\delta\in \mathcal{H}$, which implies that $\delta$ is a coset leader. Therefore, we have $d_B(\delta)=\delta=\hat{\delta}$. 

\textit{Subcase 2.} If 
$\sum\limits_{\ell=0}^{h-j_{\delta}-1}\delta_{\ell}q^{\ell-h}<\sum\limits_{\ell=h}^{h+j_{\delta}}\delta_{\ell}q^{\ell-h}.$ 
By an argument similar to that in Case~1 of the proof of Theorem~\ref{bosth},
we first conclude that any integer $\delta^{'} \in \left[\delta, \hat{\delta}\right)$ is not a coset leader. Note  that $\delta_{h}>0$ and  $\hat{\delta}= 
\sum\limits_{\ell=h}^{h+j_{\delta}}\delta_{\ell}q^{\ell}+ \sum\limits_{\ell=h}^{h+j_{\delta}}\delta_{\ell}q^{\ell-h}$ when $r_{\delta}=0$. Applying Corollary \ref{corr}, one can easily verify that $\hat{\delta}\in \mathcal{H}$. This implies that $\hat{\delta}$ is a coset leader. Consequently, we have $d_B(\mathcal{C}_{(q,m,\delta)})=\hat{\delta}$. 
\end{proof}

 An $[n,k,d]$  code over $\mathbb{F}_q$ is called \emph{optimal} if there is no $[n,k',d]$ linear code with $k' >k$, or   $[n,k,d']$  linear code with $d' > d$   over $\mathbb{F}_q$.
We provide examples of narrow-sense BCH codes discussed in this paper in Tables \ref{table2} and \ref{table3}. The dimensions 
$k$
 and Bose distances $d_B$
  of the BCH codes presented in these tables are calculated using the formulas from Theorems~\ref{odda}, \ref{bosth0}, \ref{bosth}, and \ref{odd2}, with all parameters verified using Magma.
We compare these codes with the Database and find that they are either optimal or match the parameters of the best-known linear codes.  We also utilized Magma to calculate the minimum distance of the codes. For codes whose minimum distance could not be determined by Magma within a reasonable time, we denote this with ''--". Notably, the Bose distance and minimum distance of these codes are equal.

\begin{table}[htbp]
\centering
\caption{Example of binary BCH code $\mathcal{C}_{(q,m,\delta)}$  for $m\geq 4$ and $\delta\in [2,q^{\lfloor (2m-1)/3\rfloor+1}]$}
\label{table2}
\centering
\begin{tabular}{|c|c|c|c|c|c|c|c|}
\hline
$q$ & $m$ & $\delta$ & $n$ & $k$ & $d_B$ & $d$ & Optimality \\ \hline
2	&	4	&	2 -- 3	&	15	&	11	&	3	&3	&Optimal	\\ \hline
2	&	4	&	4 -- 5	&	15	&	7	&	5	&5	&Optimal	\\ \hline
2	&	4	&	6 --7	&	15	&	5	&	7	&7	&Optimal	\\ \hline
2	&	5	&	2 -- 3	&	31	&	26	&	3	&3	&Optimal	\\ \hline
2	&	5	&	4 -- 5	&	31	&	21	&	5	&	5&Optimal	\\ \hline
2	&	5	&	6 --7	&	31	&	16	&	7	&	7&Optimal	\\ \hline
2	&	5	&	8 -- 11	&	31	&	11	&	11	&	11&Optimal	\\ \hline
2	&	5	&	12 -- 15	&	31	&	6	&	15	&	15&Optimal	\\ \hline
2	&	6	&	2 -- 3	&	63	&	57	&	3	&	3&Optimal	\\ \hline
2	&	6	&	4 -- 5	&	63	&	51	&	5	&	5&Optimal	\\ \hline
2	&	6	&	8 -- 9	&	63	&	39	&	9	&	9&Best Known	\\ \hline
2	&	6	&	10 -- 11	&	63	&	36	&	11	&	11&Best Known	\\ \hline
2	&	6	&	12 -- 13	&	63	&	30	&	13	&13	&Best Known	\\ \hline
2	&	7	&	2 -- 3	&	127	&	120	&	3	&	3&Optimal	\\ \hline
2	&	7	&	4 -- 5	&	127	&	113	&	5	&	5&Optimal	\\ \hline
2	&	7	&	6 -- 7	&	127	&	106	&	7	&	7&Best Known	\\ \hline
2	&	7	&	8 -- 9	&	127	&	99	&	9	&	9&Best Known	\\ \hline
2	&	7	&	10 -- 11	&	127	&	92	&	11	&	11&Best Known	\\ \hline
2	&	7	&	12 -- 13	&	127	&	85	&	13	&	13&Best Known	\\ \hline
2	&	7	&	14 -- 15	&	127	&	78	&	15	&15	&Best Known	\\ \hline
2	&	7	&	16 -- 19	&	127	&	71	&	19	&	19&Best Known	\\ \hline
2	&	7	&	20 -- 21	&	127	&	64	&	21	&	21&Best Known	\\ \hline
2	&	7	&	24 -- 27	&	127	&	50	&	27	&27&Best Known	\\ \hline
2	&	8	&	2 -- 3	&	254	&	247	&	3	&	3&Optimal	\\ \hline
2	&	8	&	5	&	254	&	239	&	5	&	5&Optimal	\\ \hline
2	&	8	&	6 --   7	&	254	&	231	&	7	&	7&Best Known	\\ \hline
2	&	8	&	8 -- 9	&	254	&	223	&	9	&	9&Best Known	\\ \hline
2	&	8	&	10 -- 11	&	254	&	215	&	11	&11	&Best Known	\\ \hline
2	&	8	&	12 -- 13	&	254	&	207	&	13	&	13&Best Known	\\ \hline
2	&	8	&	14 -- 15	&	254	&	199	&	15	&	15&Best Known	\\ \hline
2	&	8	&	16 -- 17	&	254	&	191	&	17	&	17&Best Known	\\ \hline
2	&	8	&	18 -- 19	&	254	&	187	&	19	&	19&Best Known	\\ \hline
2	&	8	&	20 -- 21	&	254	&	179	&	21	&	21&Best Known	\\ \hline
2	&	8	&	22 -- 23	&	254	&	171	&	23	&	23&Best Known	\\ \hline
2	&	8	&	24 -- 25	&	254	&	163	&	25	&	25&Best Known	\\ \hline
2	&	8	&	26 -- 27	&	254	&	155	&	27	&	27&Best Known	\\ \hline
2	&	8	&	28 -- 29	&	254	&	147	&	29	&	29&Best Known	\\ \hline
2	&	8	&	30 -- 31	&	254	&	139	&	31	&31	&Best Known	\\ \hline
2	&	8	&	32 -- 37	&	254	&	131	&	37	&	37&Best Known	\\ \hline
2	&	8	&	38 -- 39	&	254	&	123	&	39	&	39&Best Known	\\ \hline
2	&	8	&	40 -- 43	&	254	&	115	&	43	&43	&Best Known	\\ \hline
2	&	8	&	44 -- 45	&	254	&	107	&	45	&	45&Best Known	\\ \hline
2	&	8	&	46 -- 47	&	254	&	99	&	47	&47	&Best Known	\\ \hline
2	&	8	&	48 -- 51	&	254	&	91	&	51	&	51&Best Known	\\ \hline
2	&	8	&	52 -- 53	&	254	&	87	&	53	&	53&Best Known	\\ \hline
\end{tabular} 
\end{table}

\begin{table}[htbp]
\centering
\caption{Example of nonbinary BCH code $\mathcal{C}_{(q,m,\delta)}$ for $m\geq 4$ and $\delta\in [2,q^{\lfloor (2m-1)/3\rfloor+1}]$}
\label{table3}
\centering
\begin{tabular}{|c|c|c|c|c|c|c|c|}
\hline
$q$ & $m$ & $\delta$ & $n$ & $k$ & $d_B$& $d$ & Optimality \\ \hline
3	&	4	&	2	&	80	&	76	&	2	&	2	&	Optimal	\\ \hline
3	&	4	&	3 -- 4	&	80	&	72	&	4	&	4	&	Best Known	\\ \hline
3	&	4	&	8	&	80	&	60	&	8	&	8	&	Best Known	\\ \hline
3	&	4	&	9 -- 10	&	80	&	56	&	10	&	10	&	Best Known	\\ \hline
3	&	4	&	11	&	80	&	54	&	11	&	11	&	Best Known	\\ \hline
3	&	4	&	12 -- 13	&	80	&	50	&	13	&	13	&	Best Known	\\ \hline
3	&	4	&	14	&	80	&	46	&	14	&	14	&	Best Known	\\ \hline
3	&	5	&	2	&	242	&	237	&	2	&	2	&	Optimal	\\ \hline
3	&	5	&	3 -- 4	&	242	&	232	&	4	&	4	&	Optimal	\\ \hline
3	&	5	&	5	&	242	&	227	&	5	&	5	&	Best Known	\\ \hline
3	&	5	&	6 -- 7	&	242	&	222	&	7	&	7	&	Best Known	\\ \hline
3	&	5	&	8	&	242	&	217	&	8	&	8	&	Best Known	\\ \hline
3	&	5	&	9 -- 10	&	242	&	212	&	10	&	10	&	Best Known	\\ \hline
3	&	5	&	11	&	242	&	207	&	11	&	11	&	Best Known	\\ \hline
3	&	5	&	12 -- 13	&	242	&	202	&	13	&	13	&	Best Known	\\ \hline
3	&	5	&	14	&	242	&	197	&	14	&	14	&	Best Known	\\ \hline
3	&	5	&	15 -- 16	&	242	&	192	&	16	&	16	&	Best Known	\\ \hline
3	&	5	&	17	&	242	&	187	&	17	&	17	&	Best Known	\\ \hline
3	&	5	&	18 -- 19	&	242	&	182	&	19	&	 --	&	Best Known	\\ \hline
3	&	5	&	20	&	242	&	177	&	20	&	 --	&	Best Known	\\ \hline
3	&	5	&	21 -- 22	&	242	&	172	&	22	&	 --	&	Best Known	\\ \hline
3	&	5	&	23	&	242	&	167	&	23	&	 --	&	Best Known	\\ \hline
3	&	5	&	24 -- 25	&	242	&	162	&	25	&	 --	&	Best Known	\\ \hline
3	&	5	&	26	&	242	&	157	&	26	&	 --	&	Best Known	\\ \hline
3	&	5	&	27 -- 31	&	242	&	152	&	31	&	 --	&	Best Known	\\ \hline
3	&	5	&	32	&	242	&	147	&	32	&	 --	&	Best Known	\\ \hline
3	&	5	&	33 -- 34	&	242	&	142	&	34	&	 --	&	Best Known	\\ \hline
3	&	5	&	35	&	242	&	137	&	35	&	 --	&	Best Known	\\ \hline
3	&	5	&	36 -- 38	&	242	&	132	&	38	&	 --	&	Best Known	\\ \hline
3	&	5	&	39 -- 40	&	242	&	127	&	40	&	 --	&	Best Known	\\ \hline
3	&	5	&	41	&	242	&	122	&	41	&	 --	&	Best Known	\\ \hline
3	&	5	&	42 -- 43	&	242	&	117	&	43	&	 --	&	Best Known	\\ \hline
4	&	4	&	2	&	255	&	251	&	2	&	2	&	Best Known	\\ \hline
4	&	4	&	4 -- 5	&	255	&	243	&	5	&	5	&	Best Known	\\ \hline
4	&	4	&	6	&	255	&	239	&	6	&	6	&	Best Known	\\ \hline
4	&	4	&	8 -- 9	&	255	&	231	&	9	&	9	&	Best Known	\\ \hline
4	&	4	&	10	&	255	&	227	&	10	&	10	&	Best Known	\\ \hline
4	&	4	&	11	&	255	&	223	&	11	&	11	&	Best Known	\\ \hline
4	&	4	&	12 -- 13	&	255	&	219	&	13	&	13	&	Best Known	\\ \hline
4	&	4	&	14	&	255	&	215	&	14	&	 14	&	Best Known	\\ \hline
4	&	4	&	15	&	255	&	211	&	15	&	15	&	Best Known	\\ \hline
4	&	4	&	16 -- 17	&	255	&	207	&	17	&	 17	&	Best Known	\\ \hline
4	&	4	&	18	&	255	&	205	&	18	&	 --	&	Best Known	\\ \hline
4	&	4	&	19	&	255	&	201	&	19	&	 --	&	Best Known	\\ \hline
4	&	4	&	20 -- 21	&	255	&	197	&	21	&	 --	&	Best Known	\\ \hline
4	&	4	&	22	&	255	&	193	&	22	&	 --	&	Best Known	\\ \hline
4	&	4	&	23	&	255	&	189	&	23	&	 --	&	Best Known	\\ \hline
4	&	4	&	24 -- 25	&	255	&	185	&	25	&	 --	&	Best Known	\\ \hline
4	&	4	&	25	&	255	&	185	&	25	&	 --	&	Best Known	\\ \hline
4	&	4	&	26	&	255	&	181	&	26	&	 --	&	Best Known	\\ \hline
4	&	4	&	27	&	255	&	177	&	27	&	 --	&	Best Known	\\ \hline
4	&	4	&	28 -- 29	&	255	&	173	&	29	&	 --	&	Best Known	\\ \hline
4	&	4	&	30	&	255	&	169	&	30	&	 --	&	Best Known	\\ \hline
4	&	4	&	32 -- 34	&	255	&	161	&	34	&	 --	&	Best Known	\\ \hline
4	&	4	&	35	&	255	&	159	&	35	&	 --	&	Best Known	\\ \hline
4	&	4	&	36 -- 37	&	255	&	155	&	37	&	 --	&	Best Known	\\ \hline
\end{tabular}
\end{table}

\section{Dimension and Bose distance of  \texorpdfstring{$\mathcal{C}_{(q,m,\delta)} $ }\  with  \texorpdfstring{$\delta=aq^{h+k}+b $}. }
As a special case of our main results in the above sections,  we give the dimension and Bose distance of $\mathcal{C}_{(q,m,\delta)}$ for $m\geq 4$ and $\delta=a q^{h+k}+b$ with $k\in [m-2h, \lfloor(2m-1)/3\rfloor-h ]$, $a\in [1,q-1]$ and $b\in \left[1,q^{m-h-k}\right]$ in the following corollaries. 
\begin{Corollary}\label{example1}
   Let $m\geq 5$ be an odd integer and  
$ k\in [1,  \lfloor (2m-1)/3\rfloor-h]$ be  an integer. If 
 $\delta=a q^{h+k}+b$ for some integers  $a \in [1, q-1]$  and   $b\in [1, q^{h-k+1}]$, 
 then   
  \begin{equation}\label{cor51}
    d_B(\mathcal{C}_{(q,m,\delta)})=
    \begin{cases}
     \delta,& \hbox{if }  b> a q^{2k-1},\\
 aq^{h+k}+aq^{2k-1}+1, 
 &\hbox{if }   b\leq a q^{2k-1},
    \end{cases}
    \end{equation}
  and 
 \begin{multline}
     \label{cor52}
    \mathrm{dim}(\mathcal{C}_{(q,m,\delta)})= \\
    \begin{cases}
      \left(\begin{aligned}    
     &n-mN(\delta) \\
     &+  ma^2(q-1)q^{2k-3}\left[(q-1)k+1\right] \end{aligned}\right)
     , &\hspace{-7pt}\hbox{if }  b> a q^{2k-1},\\
 \left( \begin{aligned}    
&n-m a q^{h+k-1}(q-1)\\
&   - ma(q-1)q^{2k-2} \\
&+ ma^2 (q-1) q^{2k-3}\left[ (q-1)k+1
 \right] \end{aligned}\right), 
 &\hspace{-7pt}\hbox{if }   b\leq a q^{2k-1}.
    \end{cases}
   \end{multline}  
\end{Corollary}
\begin{proof}
By applying Theorem \ref{bosth}, we can directly derive equation (\ref{cor51}). Therefore, we only demonstrate that (\ref{cor52}) holds.
It is clear that 
$V(\delta-1)$ has the form 
\begin{equation}\notag
    (\mathbf{0}_{h-k}, a, \mathbf{0}_{2k_{}-1},\delta_{h-k},\ldots,\delta_0)
\end{equation}
with $\sum\limits_{\ell=0}^{h-k}\delta_{\ell}q^{\ell}=b-1$. 
By definition,  it follows that 
\begin{itemize}
    \item $s_{\delta}=k_{\delta}=k$; 
    \item $\mu(\delta)=\min\left\{b-1, a q^{2k-1}\right\}$;
    \item $\mathcal{T}_i(\delta)=\left[q^{k-i},a q^{k-i}-1\right]\cap \{t\in \mathbb{Z}: q\nmid t\}$ for each integer $i\in [-k+1,k]$.
\end{itemize}
Then applying Lemma  \ref{lemma10}, we can   obtain  
\begin{equation}\notag
    \sum\limits_{t\in \mathcal{T}_i(\delta)} N(tq^{2i-1}+1)= 
    \begin{cases}
        \frac{1}{2}(a^2-1)(q-1)^2q^{2k-3},\\ \quad \quad  \hbox{if }i\in [ -k+2,  k-1],\\
        \frac{1}{2} a(a-1)(q-1)q^{2k-2},\\ 
       \quad  \quad \hbox{if }i=k \hbox{ or }-k+1.
        \end{cases}       
\end{equation}
It follows that 

\begin{IEEEeqnarray}{rCl}
\IEEEeqnarraymulticol{3}{l}{\sum\limits_{i=-k_{\delta}+1}^{k_{\delta}}\sum\limits_{t\in \mathcal{T}_i(\delta)}N(tq^{2i-1}+1)} \nonumber\\ \quad 
& = & \sum\limits_{i=-k+1}^{k}\sum\limits_{t\in \mathcal{T}_i(\delta)}N(tq^{2i-1}+1) \nonumber \\
& =& a (a-1)(q-1)q^{2k-2}+(a^2-1)(k-1)(q-1)^2q^{2k-3}. \nonumber
\end{IEEEeqnarray}
By substituting the values of $k_{\delta}$, $\mu(\delta)$ and the above expression into (\ref{ff}), we obtain 
\begin{equation}\notag
     f(\delta)=
     \begin{cases}        
  a^2(q-1)q^{2k-3}\left[k(q-1)+1\right],\hbox{ if }b> a q^{2k-1},\\
  a (q-1) q^{2k-3}\left[ ak(q-1)+a-q 
 \right]+N(b),\\
 \quad \quad \hbox{if } b\leq a q^{2k-1}.
     \end{cases}
 \end{equation}
Finally,  noticing that $N(\delta)-N(b)=a q^{h+k-1}(q-1)$,  we can apply Theorem  \ref{odda} to derive the result.  
\end{proof}

\begin{Corollary}\label{example2}
 Let $m\geq 4$ be an even intger and  $0 \leq k\leq \lfloor (2m-1)/3\rfloor-h$ be an integer. If 
 $\delta=a q^{h+k}+b$ for some integers  $a \in [1, q-1]$  and   $b\in [1, q^{h-k}]$, 
 then 
 \begin{equation}\label{cor61}
   d_B(\mathcal{C}_{(q,m,\delta)})=
    \begin{cases}
      aq^h+a ,& \hbox{if }  k=0, b\leq  a, \\
\delta,
 &\hbox{if } k=0,  b>a,\\
  \delta, & \hbox{if }k\geq 1, b> aq^{2k},\\
       aq^{h+k}+aq^{2k}+1,&\hbox{if }k\geq 1, b\leq  aq^{2k}.
    \end{cases}
    \end{equation}   
and 
\begin{equation}\label{cor62}
\begin{split}
   &\mathrm{dim}(\mathcal{C}_{(q,m,\delta)})= \\
    &\begin{cases}
      n-maq^{h-1}(q-1) +\frac{1}{2}m(a-1)^2 ,& \hbox{if }  k=0,  b\leq  a, \\
n-m N(\delta)+\frac{1}{2}ma^2,
 &\hbox{if } k=0,  b>a,\\
\left(\begin{aligned} 
 & ma^2(q-1)^2q^{2k-2} \left(k-\textstyle\frac{1}{2}\right) \\
 &+ ma^2(q-1)q^{2k-1} \\&+ n-mN(\delta) 
  \end{aligned}\right),& \hbox{if }k\geq 1, b> aq^{2k}\\
     \left(\begin{aligned}       
     &  m a^2(q-1)^2q^{2k-2}\left(k-\textstyle\frac{1}{2}\right)  \\
     &+ma(a-1)(q-1)q^{2k-1}\\&+ n-maq^{h+k-1}(q-1) \end{aligned}\right), &\hbox{if }k\geq 1, b\leq  aq^{2k}.  
    \end{cases}
    \end{split}
\end{equation}
\end{Corollary}

\begin{proof}
By applying Theorem \ref{odd2}, it is easy to derive equation (\ref{cor61}). Therefore, we only demonstrate that (\ref{cor62}) holds.
    If $k=0$, then we have $\delta-1 = a q^h+b-1$ with  $0\leq b-1<q^{h-1}$. It follows that $q^h<\delta\leq q^{h+1}$ and $\delta_h=a$.  Consequenlty,  we  have $s_{\delta}=k_{\delta}=0$, which implies 
    \begin{itemize}
        \item $\widetilde{\mu}(
    \delta)=\min\{b-1, a\}$;
    \item $\tau(\delta)=1$ if $a \leq b-1$, and $\tau(\delta)=0$ if otherwise. 
    \end{itemize}
Thus, by substituting $\delta_h=a$, and  the values of $\tau(\delta)$ and $\widetilde{\mu}(\delta)$ for corresponding  case into equations (\ref{f}) and (\ref{g}), we derive  \begin{equation}\notag
    \widetilde{f}(\delta)=\begin{cases}
    \frac{1}{2}a(a-1)+N(b),& \hbox{if } k=0 \hbox{ and }b\leq a,\\
   \frac{1}{2}a(a-1)+N(a+1),& \hbox{if }k=0 \hbox{ and } b> a, 
    \end{cases}
\end{equation}
and 
\begin{equation}\notag
g(\delta)=\begin{cases}
    a-1,& \hbox{if }k=0 \hbox{ and } b\leq a,\\
  a,& \hbox{if }k=0 \hbox{ and } b> a. 
    \end{cases}
\end{equation}

If $k\geq 1$, then we have  $k_{\delta}=k$ and $V(\delta -1)$ has the form 
\begin{equation}\notag
    (\mathbf{0}_{h-k-1}, a, \mathbf{0}_{2k}, \delta_{h-k-1}, \ldots,\delta_0)
\end{equation}with $\sum\limits_{\ell=0}^{h-k}\delta_{\ell}q^{\ell}=b-1$.
By definition, it follows that 
\begin{itemize}
    \item $s_{\delta}=k_{\delta}=k$;
    \item $\tau(\delta)=0$;
    \item $\widetilde{\mu}(\delta)=\min\left\{b-1, a q^{2k}\right\}$;
   \item $\mathcal{T}_i(\delta)=\left[q^{k-i},a q^{k-i}-1\right]\cap \{t\in \mathbb{Z}: q\nmid t\}$ for each integer $i\in [-k,k]$.
\end{itemize}
Applying Lemma \ref{lemma10}, we obtain 
\begin{equation}\notag
    \sum\limits_{t\in \mathcal{T}_i(\delta)} N(tq^{2i}+1)= 
    \begin{cases}
        \frac{1}{2}(a^2-1)(q-1)^2q^{2k-2},\\ \quad  \quad \hbox{ if }i\in [-k+1, ,k-1]\setminus \{0\},\\
        \frac{1}{2} a(a-1)(q-1)q^{2k-1}, \\ \quad \quad \hbox{ if }i=k \hbox{ or }-k,\\
 \left(\begin{aligned}   
& \textstyle\frac{1}{2}(a^2-1)(q-1)^2q^{2k-2} \\
&+ \textstyle\frac{1}{2}(a-1)(q-1)q^{k-1}\end{aligned}\right), 
\\ \quad \quad \hbox{ if }i=0.
        \end{cases}       
\end{equation}
It follows that 
\begin{IEEEeqnarray}{rCl}
\sum\limits_{i=-k_{\delta}}^{k_{\delta}}\sum\limits_{t\in \mathcal{T}_i(\delta)}N(tq^{2i}+1)&=&
    \sum\limits_{i=-k}^{k}\sum\limits_{t\in \mathcal{T}_i(\delta)}N(tq^{2i}+1) \nonumber\\
    &=&(a^2-1)(q-1)^2(k-\frac{1}{2})q^{2k-2} \nonumber\\
    && +\> a(a-1)(q-1)q^{2k-1} \nonumber \\
    && + \> \frac{1}{2}(a-1)q^{k-1}(q-1). \nonumber
\end{IEEEeqnarray}
By substituting the values of $k_{\delta}$, $\widetilde{\mu}(\delta)$ and the above expression into (\ref{f}) and (\ref{g}), we obtain 
\begin{equation}\notag
    \widetilde{f}(\delta)=\begin{cases}
       \left( \begin{aligned}     
     &   a^2(q-1)^2q^{2k-2}\left(k-\textstyle\frac{1}{2}\right)\\
     &+ a^2(q-1)q^{2k-1} \\
        &+\textstyle\frac{1}{2}aq^{k-1}(q-1)\end{aligned}\right),&\hspace{-8pt} \hbox{if } k\geq 1, b> aq^{2k},\\
  \left( \begin{aligned}       
 & a^2(q-1)^2q^{2k-2} (k-\textstyle\frac{1}{2}) \\
 &+a(a-1)(q-1)q^{2k-1} \\
  &+\textstyle\frac{1}{2}aq^{k-1}(q-1)+N(b) \end{aligned} \right),  &
 \hspace{-8pt} \hbox{if } k\geq 1, 
  b\leq  aq^{2k},
     \end{cases}
\end{equation}
and 
\begin{equation}\notag
    g(\delta)=aq^{k-1}(q-1)\quad \hbox{for }k\geq 1.
\end{equation}

Finally, by applying  Theorem \ref{odda}, we obtain the result.   
\end{proof}

\section{Conclusion}
In this paper, we investigate the dimension and Bose distance of primitive BCH codes. The main contributions of this work are summarized as follows:

\begin{itemize}
\item We construct an order-preserving bijective map $V$ from the set of all integers in $[0,q^m-1]$ to ${Z}_q^m$. This mapping allows each integer $a$ to be uniquely identified with a sequence $V(a) \in {Z}_q^m$. Using this approach, for $m \geq 4$ and each integer $k \in [m-2h, \lfloor (2m-1)/3 \rfloor - h]$, we partition the integers in $\mathcal{S} \cap [q^{h+k}, q^{h+k+1})$ (for odd $m$) and $(\mathcal{S} \cup \mathcal{H}) \cap [q^{h+k}, q^{h+k+1})$ (for even $m$) into disjoint classes via Theorem~\ref{th2}. Additionally, we determine the size of the $q$-cyclotomic coset $C_a$ for $a \in [1, q^{m-\lfloor m/3 \rfloor})$ when $m$ is odd, and for $a \in [1, q^{m-\lfloor m/4 \rfloor})$ when $m$ is even.

\item We determine the dimension of narrow-sense primitive BCH codes $\mathcal{C}_{(q,m,\delta)}$ for $m \geq 4$ and $2 \leq \delta \leq q^{\lfloor (2m-1)/3 \rfloor + 1}$ by establishing two concise formulas in Theorem~\ref{odda}, thereby extending the known dimension of $\mathcal{C}_{(q,m,\delta)}$ for $m\geq 4$ and $2\leq  \delta\leq q^{\lfloor (m+1)/2\rfloor+1}$. In Theorem~\ref{ext}, these findings are   extended to some non-narrow-sense BCH codes $\mathcal{C}_{(q,m,\delta,b)}$.

\item We establish the Bose distance of BCH codes $\mathcal{C}_{(q,m,\delta)}$ for $m \geq 4$ and $2 \leq \delta \leq q^{\lfloor (2m-1)/3 \rfloor + 1}$ through Theorems~\ref{bosth0}-\ref{odd2}. 

\item Applying our results, we compute the parameters of several BCH codes, which are either optimal or have the best-known parameters.

\item As an illustration of our main results, we present the dimension and Bose distance of $\mathcal{C}_{(q,m,\delta)}$ for $m \geq 4$ and $\delta = q^{h+k} + b$ with $k \in [m-2h, \lfloor (2m-1)/3 \rfloor - h]$, $a \in [1, q-1]$, and $b \in \left[1, q^{m-h-k}\right]$ in Corollaries \ref{example1} and \ref{example2}.
\end{itemize}

\appendices
\section{Proof of Lemma \ref{le11} }
\begin{proof}
By definition, we first have 
\begin{equation}\label{l6ee1}
\sum\limits_{t=q^{k}}^{q^{k+1}-1}N(t+1)=\sum\limits_{t=q^{k}}^{q^{k+1}-1}t-\sum\limits_{t=q^{k}}^{q^{k+1}-1}\lfloor t/q\rfloor.
\end{equation}
Through direct computation, we obtain 
\begin{equation}\label{l6ee2}
\sum\limits_{t=q^{k}}^{q^{k+1}-1}t=
\begin{cases}
\frac{1}{2}(q^{2}-q),&\hbox{ if } k=0,\\
\frac{1}{2}(q^{2k+2}-q^{2k}-q^{k+1}+q^{k}), &\hbox{ if } k\geq 1.
\end{cases}
\end{equation}

Next, we compute the value of $\sum\limits_{t=q^{k}}^{q^{k+1}-1}\lfloor t/q\rfloor.$ 
 If $k=0$,  then $\lfloor t/q\rfloor=0$ for all $t\in \left[q^k,q^{k+1}-1\right]$, and hence   
$\sum\limits_{t=q^k}^{q^{k+1}-1}\lfloor t/q\rfloor =0$.  
If  $k\geq 1,$  then 
$\left[q^{k}, q^{k+1}-1\right]$ can be partitioned as $$\left[q^{k}, q^{k+1}-1\right]=\bigsqcup\limits_{i=0}^{q^{k-1}(q-1)-1}
\left[q^{k}+i q, q^{k}+(i+1)q-1\right].$$  Moreover, for each integer $i\in \left[0,q^{k-1}(q-1)-1\right]$ and any integer $t\in \left[q^{k}+i q, q^{k}+(i+1)q-1\right]$, we have   $\lfloor t/q\rfloor =q^{k-1}+i$.  Therefore,  
  \begin{equation}\notag
\begin{split}\sum\limits_{t=q^{k}}^{q^{k+1}-1}\lfloor t/q\rfloor&=\sum\limits_{i=0}^{q^{k-1}(q-1)-1}\sum\limits_{t=q^{k}+i q}^{q^k+(i+1)q-1}(q^{k-1}+i)\\
     &=\frac{1}{2}\left(q^{2k+1}-q^{2k-1}-q^{k+1}+q^k\right)
     \end{split}
 \end{equation}
  for  $k\geq 1.$  
Combining (\ref{l6ee1}), (\ref{l6ee2}) and the value of $\sum\limits_{t=q^k}^{q^{k+1}-1}\lfloor t/q\rfloor$,  we obtain the desired equality.   
\end{proof}

\section{Proof of Lemma \ref{lemma10} }
\begin{proof}
The arguments to obtain the equations (\ref{wwwww}) and (\ref{wwwww2}) are similar, so we only demonstrate here that equation (\ref{wwwww}) holds. 

If  $1\leq i\leq k$, then  we have 
$N(tq^{2i-1}+1)=tq^{2i-2}(q-1)$.
It follows that 
\begin{multline}
  \sum\limits_{t=q^{k-i},q\nmid t}^{aq^{k-i}-1}N(tq^{2i-1}+1) \\= 
  \sum\limits_{t=q^{k-i}}^{aq^{k-i}-1}tq^{2i-2}(q-1)-\sum\limits_{t=\mathcal{W}_i}^{aq^{k-i}-1}tq^{2i-2}(q-1),\nonumber
\end{multline}
where $ \mathcal{W}_i= \left[q^{k-i},  q^{k-i+1}-1\right]\cap  \{t\in \mathbb{Z}: 
 q\mid t\}$. 

 If $i=k$, then $\mathcal{W}_i=\varnothing$. Consequently, \begin{equation}\notag
     \sum\limits_{t\in \mathcal{W}_i}tq^{2i-2}(q-1)=0\quad \hbox{for }i=k.
 \end{equation}
If $1\leq i\leq k-1$, then $$\mathcal{W}_i=\{q^{k-i}+\alpha q: \alpha=0,\ldots,a^{k-i-1}(a-1)-1\}.$$ Thus, 
\begin{IEEEeqnarray}{rCl}
  \IEEEeqnarraymulticol{3}{l}{\sum\limits_{t\in \mathcal{W}_i}tq^{2i-2}(q-1)} \nonumber \\ &=&\sum\limits_{\alpha=0}^{q^{k-i-1}(a-1)-1}(q^{k-i}+\alpha q)q^{2i-2}(q-1) \nonumber \\
    &=&\frac{1}{2}(a-1)(q-1)\left[(a+1)q^{2k-3}-q^{k+i-2}\right]. \nonumber 
\end{IEEEeqnarray}
    Furthermore, one can easily obtain 
    \begin{equation}\notag
        \sum\limits_{t=q^{k-i}}^{aq^{k-i}-1}tq^{2i-2}(q-1)=\begin{cases}
            \frac{1}{2}a(a-1)q^{2k-2}(q-1),\\ \quad \quad \hbox{if }i=k,\\
         \left(\begin{aligned}
                    &\textstyle\frac{1}{2}(a^2-1)(q-1)q^{2k-2}\\ &-\textstyle\frac{1}{2}(a-1)(q-1)q^{k+i-2}\end{aligned} \right),\\ \quad \quad  \hbox{if }1\leq i\leq k-1.
        \end{cases}
    \end{equation}
Then we can conclude from the above four equations that the equality in (\ref{wwwww}) holds for $1\leq i\leq k.$

If $-k+1\leq i\leq 0$, then for  any  integer $t\in \left [q^{k-i}, a q^{k-i}-1\right]\cap \{t\in \mathbb{Z}:q\nmid t\}$  with the $q$-adic expansion
 $\sum\limits_{\ell=0}^{k-i}t_{\ell}q^{\ell}$, we have the following observations:
 \begin{itemize}
 \item $N( tq^{2i-1} +1)=N(\lfloor tq^{2i-1}\rfloor+1)$;
     \item The integer $t$ admits the unique decomposition
     $t=\lfloor tq^{2i-1}\rfloor \cdot q^{-2i+1}+\sum\limits_{\ell=0}^{-2i}t_{\ell}q^{\ell};$
 \item As $ t $ ranges over integers in  $ \left[q^{k-i}, a q^{k-i}-1\right]$  not divisible by $ q $, the value of $ \lfloor tq^{2i-1} \rfloor $ ranges over integers from  $q^{k+i-1} $ to $ a q^{k+i-1}-1$;
 \item For each fixed  $\lfloor tq^{2i-1}\rfloor$, the summation  $\sum\limits_{\ell=0}^{-2i}t_{\ell}q^{\ell}$ ranges from integers in  $[1,q^{-2i+1}-1]$  not devisble by $q$, yielding exactly $q^{-2i}(q-1)$ 
 distinct integers. 
 \end{itemize}
  Therefore, for $-k+1\leq i\leq 0$, we have 
    \begin{multline} \label{fes}
        \sum\limits_{t= q^{k-i},q\nmid t}^{a q^{k-i}-1} N(tq^{2i-1}+1)\\ = \sum\limits_{\lfloor tq^{2i-1} \rfloor=q^{k+i-1}}^{aq^{k+i-1}-1 } q^{-2i}(q-1)
N\left(\lfloor tq^{2i-1}\rfloor+1\right).   
\end{multline}
In addition, we can use a similar argument to the  proof of Lemma \ref{le11}
to obtain 
 \begin{multline} \notag
     \sum\limits_{\lfloor tq^{2i-1}\rfloor=q^{k+i-1}}^{aq^{k+i-1}-1 }   N(\lfloor tq^{2i-1}\rfloor+1)= \\
     \begin{cases}
         \frac{1}{2}(a^2-1)(q-1) q^{2k+2i-3},& \hbox{for }-k+2\leq i\leq 0,\\
         \frac{1}{2}a(a-1)q^{2k+2i-2},&\hbox{for }i=-k+1.
     \end{cases}
   \end{multline}
     Substituting this into equation (\ref{fes}), we can conclude that equation (\ref{wwwww}) holds for $-k+1\leq i\leq 0$. This completes the proof.  
\end{proof}

\section{Proof of Assertion \ref{as1}}
\begin{proof}
If  $m$ is  odd,  our first goal is to  show that 
\begin{equation}
\medmath{\left[\sum\limits_{\ell=h-k_{\delta}+1}^{h+k_{\delta}}\delta_{\ell}q^{\ell},\delta-1\right]\cap \mathcal{S}=\left[\sum\limits_{\ell=h-k_{\delta}+1}^{h+k_{\delta}}\delta_{\ell}q^{\ell},\delta-1\right]\cap \mathcal{A}_{k_{\delta}}(s_{\delta}). }\label{cl1}
\end{equation}
Recall that  $\sum\limits_{\ell=0}^{h+k_{\delta}}\delta_{\ell}q^{\ell}$ is the $q$-adic expansion of $\delta-1$. We have 
\begin{equation}\notag
V(\delta-1)=(\mathbf{0}_{h-k_{\delta}},\delta_{h+k_{\delta}},\ldots,\delta_{0})
\end{equation}
and 
\begin{equation}\notag
V\left(\sum\limits_{\ell=h-k_{\delta}+1}^{h+k_{\delta}}\!\!\delta_{\ell}q^{\ell}\right)\!=\!(\mathbf{0}_{h-k_{\delta}},\delta_{h+k_{\delta}},\ldots,\delta_{h-k_{\delta}+1}, \mathbf{0}_{h-k_{\delta}+1}). 
\end{equation}
Suppose that $a\in \left[\sum\limits_{\ell=h-k_{\delta}+1}^{h+k_{\delta}}\delta_{\ell}q^{\ell},\delta-1\right]$ is an integer with $q$-adic expansion $\sum\limits_{\ell=0}^{h+k_{\delta}}a_{\ell}q^{\ell}$. It follows that  
$V(\sum\limits_{\ell=h-k_{\delta}+1}^{h+k_{\delta}}\delta_{\ell}q^{\ell})\leq V(a)\leq V(\delta-1). $ The above two equalities imply  that 
\begin{equation}\begin{aligned}\label{vvv}
V(a)&=(\mathbf{0}_{h-k_{\delta}},a_{h+k_{\delta}},\ldots,a_0)\\ &= (\mathbf{0}_{h-k_{\delta}},\delta_{h+k_{\delta}},\ldots,\delta_{h-k_{\delta}+1} , a_{h-k_{\delta}},\ldots,a_0).
\end{aligned}\end{equation}
Recalling the definition of  $s_{\delta}$,   it follows that $s_{\delta}$ is  the smallest integer in $[-k_{\delta}+1,k_{\delta}]$ such that $a_{h+s_{\delta}}>0$.  
 By Remark  \ref{cor2},  this implies that  $a\not \in \mathcal{A}_{k_{\delta}}(i)$ for any integer $i\neq s_{\delta}$.  Therefore, 
$\left[\sum\limits_{\ell=h-k_{\delta}+1}^{h+k_{\delta}}\delta_{\ell}q^{\ell},\delta-1\right]\cap \mathcal{A}_{k_{\delta}}(i)=\varnothing$ for any integer $i\neq s_{\delta}.$  Then applying  Theorem \ref{th2},  we can conclude that  (\ref{cl1}) holds.  

Now, let us count the number of integers in the set $\left[\sum\limits_{\ell=h-k_{\delta}+1}^{h+k_{\delta}}\delta_{\ell}q^{\ell},\delta-1\right]\cap \mathcal{A}_{k_{\delta}}(s_{\delta}).$
By  the definition of $\mathcal{A}_{k_{\delta}}(s_{\delta})$ and (\ref{vvv}), we can conclude that $a\in \left[\sum\limits_{\ell=h-k_{\delta}+1}^{h+k_{\delta}}\delta_{\ell}q^{\ell},\delta-1\right]\cap \mathcal{A}_{k_{\delta}}(s_{\delta})$ if and only if 
$V(a)$ has the form 
\begin{equation}\notag
    (\mathbf{0}_{h-k_{\delta}},
\delta_{h+k_{\delta}},\ldots,\delta_{h+s_{\delta}},  \mathbf{0}_{h-k_{\delta}},
a_{k_{\delta}+s_{\delta}-1},\ldots,a_0)
\end{equation}
with  $\sum\limits_{\ell=0}^{k_{\delta}+s_{\delta}-1}a_{\ell}q^{\ell}\leq  \mu(\delta)$ and $a_0>0$.
Consequently, we have   
\begin{equation}\notag
 \medmath{ \begin{aligned}
&\left| \left[\sum\limits_{\ell=h-k_{\delta}+1}^{h+k_{\delta}}\delta_{\ell}q^{\ell},\delta-1\right]\cap \mathcal{A}_{k_{\delta}}(s_{\delta})\right|= \\
&\left|\left\{
(a_{k_{\delta}+s_{\delta}\!-\!1},\ldots,a_0)\in Z_q^{k_{\delta}+s_{\delta}} :\sum\limits_{\ell=0}^{k_{\delta}+s_{\delta}-1}a_{\ell}q^{\ell}\!\leq \! \mu(\delta), a_0\!>\!0
  \right\}\right|.  
\end{aligned}}\end{equation}
Noticing that $\mu(\delta)< q^{k_{\delta}+s_{\delta}}$, we apply Lemma  \ref{le5} to obtain 
\begin{equation}\notag
 \left| \left[\sum\limits_{\ell=h-k_{\delta}+1}^{h+k_{\delta}}\delta_{\ell}q^{\ell},\delta-1\right]\cap \mathcal{A}_{k_{\delta}}(s_{\delta})\right |= N(\mu(\delta)+1).
\end{equation} Recalling the equality in (\ref{cl1}), it follows that equation (\ref{as1-1}) holds.

If $ m $ is even, we can apply a similar method as above to derive equation (\ref{as1e1}). Next, we will establish equation (\ref{ww1}).
We can employ an analogous argument to that used in deriving (\ref{vvv}) to conclude that $ V(a) $ takes the form
\begin{equation}\notag
(\mathbf{0}_{h-k_{\delta}-1},\delta_{h+k_{\delta}},\ldots,\delta_h, a_{h-1},\ldots,a_0)
\end{equation} for any integer $a\in \left[\sum\limits_{\ell=h}^{h+k_{\delta}}\delta_{\ell}q^{\ell}, \delta-1\right]. $
Applying Corollary \ref{corr}, it follows that  an integer 
$a\in \left[\sum\limits_{\ell=h}^{h+k_{\delta}}\delta_{\ell}q^{\ell}, \delta-1\right]\cap  \mathcal{H}$ if and only if $V(a)$ has the form 
\begin{equation}\notag
(\mathbf{0}_{h-k_{\delta}-1},\delta_{h+k_{\delta}},\ldots,\delta_h, \mathbf{0}_{h-k_{\delta}-1},\delta_{h+k_{\delta}}, \ldots,\delta_h)
\end{equation} with   $\delta_h>0$ and $\sum\limits_{\ell=h}^{h+k_{\delta}}\delta_{\ell}q^{\ell-h}\leq \sum\limits_{\ell=0}^{h-1}\delta_{\ell}q^{\ell}$. 
This implies that equation (\ref{ww1}) holds. 
\end{proof}

\section{Proof of Assertion \ref{as2}}
\begin{proof}
If $m$ is odd, we have $k_{\delta}\geq m-2h=1$. By the definition of $\mathcal{A}_{k_{\delta}}(i)$, we can conclude that an integer 
 $a\in \left [q^{h+k_{\delta}},  \sum\limits_{\ell=h-k_{\delta}+1}^{h+k_{\delta}}\delta_{\ell}q^{\ell}\right)\cap \mathcal{A}_{k_{\delta}}(i)$  
if and only if  
\begin{equation}\label{ine}
    V(q^{h+k_{\delta}})\leq V(a)<V(\sum\limits_{\ell=h-k_{\delta}+1}^{h+k_{\delta}} \delta_{\ell} q^{\ell}),
    \end{equation}
and 
$V(a)$ has form specified in (\ref{p33}) with $k=k_{\delta}$
while satisfying the conditions in (\ref{p31}) and (\ref{p32}). 
Notice that 
\begin{equation}\notag
V(q^{h+k_{\delta}})=(\mathbf{0}_{h-k_{\delta}}, 1,\mathbf{0}_{h+k_{\delta}})
\end{equation}
and
\begin{equation}\notag
V(\sum\limits_{\ell=h-k_{\delta}+1}^{h+k_{\delta}} \delta_{\ell} q^{\ell})=(\mathbf{0}_{h-k_{\delta}}, \delta_{h+k_{\delta}},\ldots,\delta_{h-k_{\delta}+1},\mathbf{0}_{h-k_{\delta}+1}).
\end{equation}
The form of $V(a)$ in (\ref{p33}) and the inequality  $a_0>0$ in (\ref{p32}) imply that the  inequality  in (\ref{ine}) holds  if and only if 
\begin{equation}\notag
    \begin{aligned}
      (1,\mathbf{0}_{2k_{\delta}-1}) & \leq (a_{h+k_{\delta}}, \ldots,a_{h+i}, \mathbf{0}_{k_{\delta}+i-1}) \\ 
                                   &<(\delta_{h+k_{\delta}}, \ldots,\delta_{h-k_{\delta}+1}).
    \end{aligned}
  \end{equation}
Since $s_{\delta}$ is the smallest integer in $[-k_{\delta}+1,k_{\delta}]$ such that $\delta_{h+s_{\delta}}>0$, this is 
 further equivalent to 
\begin{equation}\notag
    q^{k_{\delta}-i}\leq \sum\limits_{\ell=h+i}^{h+k_{\delta}}a_{\ell}q^{\ell-(h+i)}< \sum\limits_{\ell=h+s_{\delta}}^{h+k_{\delta}}\delta_{\ell}q^{\ell-h-i}.
\end{equation}   
Recall the definition of the set $\mathcal{T}_i(\delta)$ and Remarks \ref{rm3} and \ref{rm4}. We   
 can  conclude 
  that an integer $a\in \left[q^{h+k_{\delta}},  \sum\limits_{\ell=h-k_{\delta}+1}^{h+k_{\delta}}\delta_{\ell}q^{\ell}\right)\cap \mathcal{A}_{k_{\delta}}(i)$ 
if and only if $V(a)$ has the form  specified  in (\ref{p33}) with $k=k_{\delta}$ and \begin{numcases}{}
\sum\limits_{\ell=h+i}^{h+k_{\delta}}a_{\ell}q^{\ell-(h+i)}\in \mathcal{T}_i(\delta);\label{78}\\
\sum\limits_{\ell=0}^{k_{\delta}+i-1}a_{\ell}q^{\ell}\leq \sum\limits_{h+i}^{h+k_{\delta}}a_{\ell}q^{\ell-(h+i)}\cdot q^{2i-1} \hbox{ and } a_0>0.\label{nn2}
\end{numcases}

To enumerate such integers, we have the following observations:
\begin{itemize}
    \item Each sequence $(a_{h+k_{\delta}}, \ldots, a_{h+i})$ satisfying (\ref{78}) corresponds to one integer 
    $t = \sum\limits_{\ell=h+i}^{h+k_{\delta}} a_{\ell} q^{\ell-(h+i)}$ in the set  $\mathcal{T}_i(\delta)$, and vice versa. 
\item     For each fixed $(a_{h+k_{\delta}}, \ldots, a_{h+i})$,  Lemma \ref{le5} yields that there exist exactly $N(t q^{2i-1} + 1)$ sequences $(a_{k_{\delta}+i-1}, \ldots, a_0)$ meeting the requirement in (\ref{nn2}), where $t=\sum\limits_{\ell=h+i}^{h+k_{\delta}} a_{\ell} q^{\ell-(h+i)}$. 
\end{itemize}
Therefore, we can conclude that 
for any integer $i\in [-k_{\delta}+1,k_{\delta}]$. 
\begin{equation}\notag
 \left|\left[q^{h+k_{\delta}}, \sum\limits_{\ell=h-k_{\delta}+1}^{h+k_{\delta}}\delta_{\ell}q^{\ell}\right)\cap \mathcal{A}_{k_\delta}(i)\right|= \sum\limits_{t\in \mathcal{T}_i(\delta)}N(tq^{2i-1}+1). 
 \end{equation}  
 Then by Theorem \ref{th2}, we obtain the equality in  (\ref{xx2}).   
 
If  $m$ is an even integer,  then $k_{\delta}\geq m-2h=0$.   The equality in (\ref{aas31}) can be obtained using a similar argument as presented above. Next, we establish equation (\ref{aas32}). 

By applying Corollary \ref{th1}, we can conclude that an integer  $a\in \left[q^{h+k_{\delta}}, \sum\limits_{\ell=h}^{h+k_{\delta}}\delta_{\ell}q^{\ell}\right)\cap \mathcal{H}$ if and only if 
\begin{equation}\label{z1}
  V(q^{h+k_{\delta}})\leq V(a)<V(\sum\limits_{\ell=h}^{h+k_{\delta}}\delta_{\ell}q^{\ell})
\end{equation}
and 
$V(a)$ has the form specified in (\ref{coeq1}) with $k=k_{\delta}$ and $a_{0}>0$.  
 Noticing that  $    V(q^{h+k_{\delta}})= (\mathbf{0}_{h-k_{\delta}},1,\mathbf{0}_{h+k_{\delta}})$ and $V(\sum\limits_{\ell=h}^{h+k_{\delta}}\delta_{\ell}q^{\ell})=(\mathbf{0}_{h-k_{\delta}}, \delta_{h+k_{\delta}}, \ldots,\delta_h, \mathbf{0}_h),$ it follows  that the inequality in (\ref{z1}) holds if and only if 
   \begin{equation}\notag
   (1,\mathbf{0}_{k_{\delta}})\leq (a_{k_{\delta}}, \ldots,a_0)<(\delta_{h+k_{\delta}}, \ldots,\delta_h).
   \end{equation}
   This can be equivalently represented as \begin{equation}\notag
q^{k_{\delta}}\leq \sum\limits_{\ell=0}^{k_{\delta}}a_{\ell}q^{\ell}<\sum\limits_{\ell=h}^{h+k_{\delta}}\delta_{\ell}q^{\ell-h}.
   \end{equation}
  Consequently, we have 
  \begin{equation}\notag
    \begin{aligned}
    &\left| \left[q^{h+k_{\delta}}, \sum\limits_{\ell=h}^{h+k_{\delta}}\delta_{\ell}q^{\ell}\right)\cap \mathcal{H}\right|= \\
    &
    \left| \left\{
      \begin{aligned}
        (a_{k_{\delta}}, \ldots, a_0):  q^{k_{\delta}}\leq \sum\limits_{\ell=0}^{k_{\delta}}a_{\ell}q^{\ell}<\sum\limits_{\ell=h}^{h+k_{\delta}}\delta_{\ell}q^{\ell-h}, a_0>0.
      \end{aligned}
      \right\}\right|
     \end{aligned}
   \end{equation}
  By Lemma \ref{le5}, it follows that equation (\ref{aas32}) holds.    
\end{proof}

\section{Proof of Assertion \ref{as3}}
\begin{proof}
The arguments to obtain (\ref{10-0}) and (\ref{2-10-0})
are similar, so we here only demonstrate that (\ref{10-0}) holds if $m$ is odd.  
We can use the same approach as in the proof of Assertion \ref{as2} to conclude  that an integer $a\in \mathcal{A}_{k}(i)$ if and only if $V(a)$ has the form specified  in (\ref{p33}) with 
 \begin{numcases}{}
$$\sum\limits_{\ell=h+i}^{h+k}\!\!a_{\ell}
q^{l-(h+i)}\in \left[q^{k-i}, q^{k -i+1}\!-\!1\right]\!\cap\! \{t\in \mathbb{Z}: q\nmid t\};$$ \label{as4c4}\\
$$\sum\limits_{\ell=0}^{k+i-1}a_{\ell}q^{\ell}\leq \sum\limits_{\ell=h+i}^{h+k}a_{\ell}q^{\ell-(h+i)}\cdot q^{2i-1}\hbox{ and }a_0>0.$$ \label{as4c5}
 \end{numcases}
Furthermore, we have the following observations:
\begin{itemize}

\item Each sequence $(a_{h+k},\ldots,a_{h+i})$ satisfying (\ref{as4c4}) corresponds to one integer $t=\sum\limits_{\ell=h+i}^{h+k}a_{\ell}q^{\ell-(h+i)}$ in the set  $[q^{k-i}, q^{k-i+1}-1] \cap \{t \in \mathbb{Z} \mid q \nmid t\}$, and vice versa.  
    \item For each fixed sequence $(a_{h+k}, \ldots, a_{h+i})$,  we can utilize Lemma \ref{le5}  to conclude that there exist   
$N(t q^{2i-1}+1)$ sequences $(a_{k+i-1}, \ldots, a_0)\in Z_q^{k+i}$ satisfying (\ref{as4c5}), where $t=\sum\limits_{\ell=h+i}^{h+k}a_{\ell}q^{\ell-(h+i)}$. 
\end{itemize}
Consequently,  we can conclude that
\begin{equation}\notag
    \left| \mathcal{A}_k(i)\right|=  \sum\limits_{t=q^{k-i},q\nmid t}^{q^{k-i+1}-1}N(tq^{2i-1}+1).
\end{equation}
Then applying Lemma \ref{lemma10} with  $a=q$, we derive (\ref{10-0}). 
\end{proof}

\section{Proof of Assertion \ref{as5}}
\begin{proof}
By the definition of $\mathcal{B}_k(0)$,    we can use a similar argument as in the proof of Assertion \ref{as2} to derive 
\begin{equation}\notag
\left|\mathcal{B}_k(0)\right|=\sum\limits_{t=q^k, q\nmid t}^{q^{k+1}-1}N(t+1). 
\end{equation}
If $k\geq 1$, 
by substituting  $i=0$ and $a=q$ into equation (\ref{wwwww2}) in  Lemma \ref{lemma10}, we obtain 
\begin{IEEEeqnarray}{rCl}
\left|\mathcal{B}_k(0)\right| 
= \frac{1}{2}(q-1)^2(q^{2k}-q^{2k-2}+q^{k-1}).\nonumber
\end{IEEEeqnarray}
If $k=0$, then  we have 
\begin{equation}\notag
   \left| \mathcal{B}_k(0)\right|=\sum\limits_{t=1}^{q-1}N(t+1)=\frac{1}{2}q(q-1). 
\end{equation}
This completes the proof of Assertion \ref{as5}.
\end{proof}




\ifCLASSOPTIONcaptionsoff
  \newpage
\fi

\newpage
\begin{IEEEbiographynophoto}{Run Zheng}
is currently pursuing his Ph.D. in the Department of Applied Mathematics at The Hong Kong Polytechnic University. He received his M.Phil. from The Hong Kong Polytechnic University, Hong Kong, in 2021 and his B.Sc. from Hunan University, China, in 2019. His research interests include coding theory, matrix theory, and quantum information.
\end{IEEEbiographynophoto}

\begin{IEEEbiographynophoto}{
 Nung-Sing Sze} received his B.Sc., M.Phil., and Ph.D. degrees in Mathematics from The University of Hong Kong in 2000, 2002, and 2005, respectively. From 2006 to 2009, he was a Postdoctoral Fellow at the University of Connecticut, USA. In August 2009, he joined The Hong Kong Polytechnic University, where he is currently an Associate Professor. His current research interests lie in quantum computation and quantum information science, with a particular focus on related mathematical problems in matrix and operator theory.
\end{IEEEbiographynophoto}

\begin{IEEEbiographynophoto}{Zejun Huang}
 received his Ph.D. in 2011 from East China Normal University in Shanghai, China. From July 2011 to December 2013, he was a Postdoctoral Fellow in the Department of Applied Mathematics at The Hong Kong Polytechnic University, Hong Kong. Between January 2014 and June 2019, he worked at the Institute of Mathematics, Hunan University in Hunan, China. Since July 2019, he has been with the School of Mathematical Sciences at Shenzhen University in Shenzhen, China, where he is currently a Professor. His research interests include matrix theory, graph theory, and coding theory.
\end{IEEEbiographynophoto}

\end{document}